\documentclass[11pt, draftcls, onecolumn]{IEEEtran}
\usepackage[english]{babel}
\usepackage{amsmath,amsthm}
\usepackage{amsfonts}
\usepackage{graphicx}
\usepackage{verbatim}
\usepackage{amssymb}
\usepackage{epstopdf}
\usepackage{hyperref}
\usepackage{bbm}
\usepackage{cite,bbm,graphicx,amsmath,amssymb,mathrsfs,epsf} \usepackage{multirow}
\usepackage{subfigure,cite,graphicx,epsfig,amsmath,amssymb,mathrsfs,epsf,graphics,color,enumerate}

\allowdisplaybreaks[4]

\newtheorem{theorem}{Theorem} 
\newtheorem{lemma}[theorem]{Lemma}

\newtheorem{corollary}[theorem]{Corollary}
\newtheorem{definition}{Definition}

\newtheorem{remark}{Remark}

\newcommand{\bbR}{\mathbb{R}}
\newcommand{\bbN}{\mathbb{N}}
\newcommand{\eps}{\varepsilon}
\newcommand{\tils}{\tilde{s}}

\def \E{\operatorname{E}}

\def \var{\operatorname{Var}}

\begin{document}
\title{Streaming Data Transmission in the Moderate Deviations and Central Limit Regimes}%

%
\author{\IEEEauthorblockN{Si-Hyeon Lee, Vincent Y. F. Tan, and Ashish Khisti}\thanks{S.-H. Lee and A. Khisti are with the Department of Electrical and Computer Engineering,
University of Toronto, Toronto, Canada (e-mail: sihyeon.lee@utoronto.ca; akhisti@comm.utoronto.ca).  V.~Y.~F.~Tan is with the Department of Electrical and Computer Engineering
and the Department of Mathematics, National University of Singapore,
Singapore (e-mail: vtan@nus.edu.sg). The work of V. Y. F. Tan is supported in part by a Singapore Ministry of Education (MOE) Tier 2 grant (R-263-000-B61-112).}}

\maketitle

\begin{abstract}
We consider streaming data transmission over a discrete memoryless channel. A new message is given to the encoder at the beginning of each block and the decoder decodes each message sequentially, after a delay of $T$ blocks. In this streaming setup, we study the fundamental interplay between the rate and error probability in the central limit and moderate deviations regimes and show that i) in the moderate deviations regime, the moderate deviations constant improves over the block coding or non-streaming setup by a factor of $T$ and ii) in the central limit regime, the second-order coding rate improves by a factor of approximately $\sqrt{T}$ for a wide range of channel parameters. For both regimes, we propose coding techniques that incorporate a joint encoding of fresh and previous messages. In particular, for the central limit regime, we propose a coding technique with truncated memory to ensure that a summation of constants, which arises as a result of applications of the central limit theorem, does not diverge in the error analysis. 

Furthermore, we explore interesting variants of the basic streaming setup in the moderate deviations regime.  We first consider a scenario with an erasure option at the decoder and show that both the exponents of the total error and the undetected error probabilities improve by factors of $T$. Next, by utilizing the erasure option, we show that the exponent of the total error probability can be improved to that of the undetected error probability (in the order sense) at the expense of a variable decoding delay. Finally, we also extend our results to the case where the message rate is not fixed but alternates between two values. 

\end{abstract}

\section{Introduction}
In many multimedia applications, a stream of data packets is required to be sequentially encoded and decoded under strict latency constraints. For such a streaming setup,  both the fundamental limits and optimal schemes can differ from classical communication systems. In recent years, there has been a growing interest in the characterization of fundamental limits for streaming data transmission \cite{Schulman:96,sahai_thesis, SukhavasiHassibi:11,KhistiDraper:14,DraperKhisti:11,DraperChangSahai:14}. 
In \cite{Schulman:96,sahai_thesis, SukhavasiHassibi:11}, coding techniques based on tree codes were proposed for streaming setup with applications to control systems. In \cite{KhistiDraper:14}, Khisti and Draper established the optimal diversity-multiplexing tradeoff (DMT) for streaming  over a block-fading multiple-input multiple-output channel. In \cite{DraperKhisti:11}, the same authors proposed a coding technique using finite memory for streaming over discrete memoryless channels (DMCs) that attains the same reliability as previously known semi-infinite coding techniques with growing memory. In \cite{DraperChangSahai:14}, the error exponent was studied in a streaming setup of distributed source coding. 
We note that these prior works assumed that the code operates in the large deviations regime in which the rate is bounded away from capacity (or the rate pair is strictly inside the optimal rate region for compression problems) and the error probability decays exponentially as the blocklength increases.

Other interesting asymptotic  regimes include the central limit and moderate deviations regimes. Let $n$ denote the blocklength of a single message henceforth. In the central limit regime,  the rate approaches to the capacity at a speed proportional to $\frac{1}{\sqrt{n}}$ and the error probability does not vanish as the blocklength increases. In the moderate deviations regime, the rate approaches to the capacity strictly slower than $\frac{1}{\sqrt{n}}$ and the error probability decays sub-exponentially fast as the blocklength increases. For block coding problems, both regimes have 
 received a fair amount of  attention recently. These works aim to  characterize the fundamental interplay between the coding rate and error probability.   The most notable early work on channel coding  in the central limit regime (also known as second-order asymptotics or the normal approximation regime) is that of Strassen \cite{Strassen}, who
considered DMCs and showed that the backoff from capacity scales as $\sqrt{n}$ when the error probability is fixed. Strassen also deduced the constant of proportionality, which is related to the so-called {\em dispersion}~\cite{PolyanskiyPoorVerdu:10}.   Hayashi \cite{Hayashi09} considered DMCs with cost constraints as well as discrete channels with Markovian memory. 
Polyanskiy {\em et al.} \cite{PolyanskiyPoorVerdu:10}   refined the asymptotic expansions and also compared the normal approximation to the finite blocklength (non-asymptotic) fundamental limits. For a review and extensions to multi-terminal models, the reader is referred to \cite{tan2015asymptotic}. For the moderate deviations regime,  He {\em et al.}~\cite{he2009redundancy} considered fixed-to-variable length source coding with decoder side information.   Altu\u{g} and Wagner~\cite{AltugWagner:14} initiated the study of moderate deviations for channel coding, specifically DMCs. Polyanskiy and Verd\'u~\cite{PolyanskiyVerdu:10} relaxed some assumptions in the conference version of Altu\u{g} and Wagner's work~\cite{altug2010moderate} and they also considered moderate deviations for additive white Gaussian noise (AWGN) channels. 
However, this line of research has not been extensively  studied for the streaming setup. To the best of our knowledge, there has been no prior work on the streaming setup in the moderate deviations and central limit regimes with the exception \cite{LinTanMotani:arxiv15} where the focus is on source coding. 

In this paper, we study streaming data transmission over a DMC  in the moderate deviations and central limit regimes. Our streaming setup is illustrated in Fig. \ref{fig:streaming}. In each block of length $n$, a new message is given to the encoder at the beginning, and the encoder generates a codeword as a function of all the past and current messages and transmits it over the channel. The decoder, given all the past received channel output sequences, decodes each message after a delay of $T$ blocks. 
This streaming setup introduces a new dimension not present in the block coding problems studied previously. In the special case of $T=1$, the setup reduces to the block channel coding problem. If $T\geq 2$, however, there exists an inherent tension in whether we utilize a block only for the fresh message or use it also for the previous messages with earlier deadlines. It is not difficult to see that due to the memoryless nature of the model, a time sharing scheme\footnote{In a time sharing scheme, some fraction of a block is used for a fresh message and some other fraction of the block is used for previous messages.} will not provide any gain compared to the case of $T=1$. A natural question is whether a joint encoding of fresh and previous messages would improve the performance when $T\geq 2$. 
 
\begin{figure*}[t]
 \centering
  {
  \includegraphics[width=120mm]{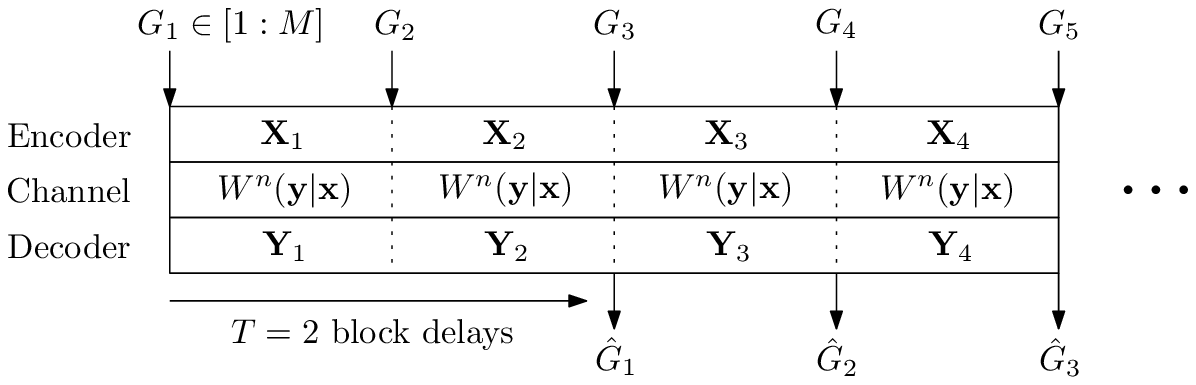}}
  \caption{Our streaming setup is illustrated for the case with $T=2$. 
  In each block, a new message is given to the encoder in the beginning and the encoder generates a codeword as a function of all the past and current messages and transmits it over the channel. Since $T=2$, the decoder decodes each message after two blocks, as a function of all the past received channel output sequences. } \label{fig:streaming}
\end{figure*} 

Our results indicate that the fundamental interplay between the rate and error probability can be greatly improved when delay is allowed in the streaming setup. 
In the moderate deviations regime, the moderate deviations constant is shown to improve over the block coding or non-streaming setup by a factor of $T$. In the central limit regime, the second-order coding rate is shown to  improve by a factor of approximately $\sqrt{T}$ for a wide range of channel parameters. For both asymptotic regimes, we propose coding techniques that incorporate a joint encoding  of fresh and previous messages. For the moderate deviations regime,  we propose a coding technique in which, for every block, the encoder jointly encodes all the previous and fresh messages and the decoder re-decodes all the previous messages in addition to the current target message. For the error analysis of this coding technique, we develop a refined and non-asymptotic version of the moderate deviations upper bound in \cite[Theorem 3.7.1]{DemboZeitouni:09} that allows us to uniformly bound the error probabilities associated with the previous messages. On the other hand, for the central limit regime,  we cannot apply such a coding technique whose memory is linear in the block index. In the error analysis in the central limit regime,  we encounter a  summation of constants as a result of applications of the central limit theorem. If the memory is linear in the block index, this summation causes the upper bound on the error probability to diverge as the block index tends to infinity. Hence, for the central limit regime, we propose a coding technique with {\em truncated} memory where the memory at the encoder varies in a periodic fashion. Our proposed construction judiciously balances the rate penalty imposed due to the truncation and the growth in the error probability due to the contribution from previous messages. By analyzing the second-order coding rate of our proposed setup, we conclude that the channel dispersion parameter also decreases approximately by a factor of $T$ for a wide range of channel parameters.

Furthermore, we explore interesting variants of the basic streaming setup in the moderate deviations regime. First, we consider a scenario where there is an erasure option at the decoder and analyze the undetected error and the total error probabilities, extending a result by Hayashi and Tan~\cite{HayashiTan:15}. Next, by utilizing the erasure option, we analyze the rate of decay of the error probability 
when a variable decoding delay is allowed. We show that such a flexibility in the decoding delay can dramatically improve the error probability in the streaming setup. This result is the analog of the classical results on variable-length decoding (see e.g., \cite{Forney:68}) to the streaming setup.
Finally, as a simple example for the case where the message rates are not constant, we consider a scenario where the rate of the messages in odd block indices and the rate of the messages in even block indices are different and analyze the moderate deviations constants separately for the two types of messages. This setting finds applications in video and audio coding where streams of data packets do not necessarily have a constant rate. 

The rest of this paper is organized as follows. In Section \ref{sec:model}, we formally state our streaming setup. The main theorems are presented in Section \ref{sec:main} and proved in Section \ref{sec:proof}. In Section \ref{sec:extension}, the moderate deviations result for the basic streaming setup is extended in various directions. We conclude this paper in Section \ref{sec:conclusion}.

\subsection{Notation}
The following notation is used throughout the paper.
We reserve bold-font for vectors whose lengths are the same as blocklength $n$. For two integers $i$ and $j$, $[i:j]$ denotes the set $\{i,i+1,\cdots, j\}$. For constants $x_1,\cdots, x_k$ and $S\subseteq [1:k]$, $x_S$  denotes the  vector $(x_j: j\in S)$ and $x^j_i$ denotes $x_{[i:j]}$ where the subscript is omitted when $i=1$, i.e., $x^j=x_{[1:j]}$. 
This notation is naturally extended for vectors $\mathbf{x}_1,\cdots, \mathbf{x}_k$, random variables $X_1,\cdots, X_k$, and random vectors $\mathbf{X}_1, \cdots, \mathbf{X}_k$.  $\mathbbm{1}\{\mathcal{E}\}$ for an event $\mathcal{E}$ denotes the indicator function, i.e., it is 1 if  $\mathcal{E}$ is true and 0 otherwise. $\lceil\cdot\rceil$ and $\lfloor\cdot\rfloor$ denote the ceiling and floor functions, respectively. 

For a DMC $(\mathcal{X},\mathcal{Y}, \{W(y|x): x\in \mathcal{X}, y\in \mathcal{Y}\} )$ and an input distribution $P$, we use the following standard notation and terminology in information theory: 
\begin{itemize}
\item Information density:
\begin{align}
i(x;y):=\log \frac{W(y|x)}{PW(y)}, \label{eqn:info_dst}
\end{align}
where $PW(y):=\sum_{x\in \mathcal{X}}P(x)W(y|x)$ denotes the output distribution. We note that $i(x;y)$ depends on $P$ and $W$ but this dependence is suppressed. The definition \eqref{eqn:info_dst} can be generalized  for two vectors $x^l$ and $y^l$ of length $l$ as follows:  
\begin{align}
i(x^l;y^l):=\sum_{j=1}^l i(x_j;y_j).
\end{align}

\item Mutual information:
\begin{align}
I(P,W)&:=\E[i(X;Y)]\\
&=\sum_{x\in \mathcal{X}}\sum_{y\in \mathcal{Y}} P(x)W(y|x)\log\frac{W(y|x)}{PW(y)}.
\end{align}

\item Unconditional information variance:
\begin{align}
U(P,W)&:=\var[i(X;Y)].
\end{align}

\item Conditional information variance: 
\begin{align}
V(P,W)&:=\E[\var[i(X;Y)|X]].
\end{align} 

\item Capacity: 
\begin{align}
C=C(W):=\max_{P\in \mathcal{P}} I(P,W),
\end{align}
where $\mathcal{P}$ denotes the probability simplex on $\bbR^{|\mathcal{X}|}$.

\item Set of capacity-achieving input distributions: 
\begin{align}
\Pi=\Pi(W):=\{P\in \mathcal{P}: I(P,W)=C(W)\}.
\end{align}

\item Channel dispersion 
\begin{align}
V=V(W)&:=\min_{P\in \Pi} V(P,W) \label{eqn:dispersion}\\
&\overset{(a)}{=}\min_{P\in \Pi} U(P,W),
\end{align}
where $(a)$ is from \cite[Lemma 62]{PolyanskiyPoorVerdu:10}, where it is shown that $V(P,W)=U(P,W)$ for all $P\in \Pi$. 
\end{itemize}

\section{Model} \label{sec:model}
Consider a DMC $(\mathcal{X},\mathcal{Y}, \{W(y|x): x\in \mathcal{X}, y\in \mathcal{Y}\} )$. A streaming code is defined as follows:
\begin{definition}[Streaming code]	\label{def:basic}
An $(n,M,\epsilon,T)$-streaming code consists of 
\begin{itemize}
\item a sequence of messages $\{G_k\}_{k\geq 1}$ each distributed uniformly over $\mathcal{G}:= [1:M]$,
\item a sequence of encoding functions $\phi_k: \mathcal{G}^k\rightarrow \mathcal{X}^n$ that maps the message sequence $G^k\in  \mathcal{G}^k$ to the channel input codeword $\mathbf{X}_k \in \mathcal{X}^n$, and
\item a sequence of decoding functions $\psi_k: \mathcal{Y}^{(k+T-1)n}\rightarrow \mathcal{G}$ that maps the channel output sequences $\mathbf{Y}^{k+T-1} \in \mathcal{Y}^{(k+T-1)n}$ to a message estimate $\hat{G}_k\in \mathcal{G}$,
\end{itemize}
that satisfies  
\begin{align}
\limsup_{N\rightarrow \infty}\sum_{k=1}^N\frac{\Pr(\hat{G}_k\neq G_k)}{N} \leq \epsilon,
\end{align}
i.e., the probability of error averaged over all block messages does not exceed $\epsilon$. 
\end{definition}

We note that a streaming code with a {\em fixed} blocklength $n$ consists of a {\em sequence} of encoding and decoding functions since a stream of messages is sequentially encoded and decoded. Fig. \ref{fig:streaming} illustrates our streaming setup for the case with $T=2$. In the beginning of block $k\in \bbN$, new message $G_k$ is given to the encoder. The encoder generates a codeword $\mathbf{X}_k$ as a function of all the past and current messages $G^k$ and transmits it over the channel in block $k$. Since $T=2$, the decoder decodes message $G_k$ at the end of block $k+1$, as a function of all the past received channel output sequences $\mathbf{Y}^{k+1}$. 



\section{Main Results} \label{sec:main}
In this section, we state our main results. The following two theorems present achievability bounds for the moderate deviations and the central limit regimes, respectively, which are proved in Section \ref{sec:proof}.
\begin{theorem}[Moderate deviations regime] \label{thm:MD}
Consider a DMC $(\mathcal{X},\mathcal{Y}, \{W(y|x): x\in \mathcal{X}, y\in \mathcal{Y}\} )$ with $V>0$ and any sequence of integers $M_n$ such that $\log M_n=nC-n\rho_n$, where $\rho_n>0, \rho_n\rightarrow 0$ and $n\rho_n^2\rightarrow \infty$.\footnote{Throughput the paper, we ignore integer constraints on the number of codewords $M_n$.} Then, there exists a sequence of $(n, M_n, \epsilon_n, T)$-streaming codes such that\footnote{If $\limsup_{n\rightarrow \infty} \frac{1}{n\rho_n^2} \log \epsilon_n\leq -\frac{1}{2\nu}$ for some $\nu>0$, $\nu$ corresponds to an upper bound on the moderate deviations constant. In the special case of $T=1$, the moderate deviations constant is shown to be the channel dispersion $V$  in \cite{AltugWagner:14,PolyanskiyVerdu:10}.} 
\begin{align}
\limsup_{n\rightarrow \infty} \frac{1}{n\rho_n^2} \log \epsilon_n \leq -\frac{T}{2V}. \label{eqn:MD_positiveV}
\end{align}
\end{theorem}

\begin{theorem}[Central limit regime] \label{thm:CL}
Consider a DMC $(\mathcal{X},\mathcal{Y}, \{W(y|x): x\in \mathcal{X}, y\in \mathcal{Y}\} )$ with $V>0$. For any $L>0$ and $0<\delta<1/2$, there exists a sequence of $(n, M_n, \epsilon_n, T)$-streaming codes such that\footnote{$L$ is termed second-order coding rate in this paper. This is slightly different from what is common in the literature where instead  $-L$ is known as the second-order coding rate \cite{Hayashi09}.  } 
\begin{align}
\log M_n&= nC-L\sqrt{n}+O(n^{\delta}\log n) \label{eqn:CL_rate}
\end{align}
and
\begin{align}
\epsilon_n&\leq \sum_{j=T}^{\infty} Q\left(\frac{\sqrt{j}}{\sqrt{V}}L\right)+O\left(n^{-\delta/2}\right). \label{eqn:CL_err}
\end{align}
\end{theorem}

The following corollary, whose proof is in Appendix  \ref{appendix:asymptotic}, elucidates a closed-form and interpretable expression for the upper bound on the error probability in \eqref{eqn:CL_err}.
\begin{corollary}\label{coro:CL_asymp}
Consider a DMC $(\mathcal{X},\mathcal{Y}, \{W(y|x): x\in \mathcal{X}, y\in \mathcal{Y}\} )$ with $V>0$. 
For any $L>0$, there exists a sequence of $(n, M_n, \epsilon_n, T)$-streaming codes such that 
\begin{align}
\lim_{n\rightarrow \infty}\frac{nC-\log M_n}{\sqrt{n}}&= L
\end{align}
and
\begin{align}
\limsup_{n\rightarrow \infty}\epsilon_n&\leq  c_{L,V,T}Q\left(\sqrt{\frac{T}{V}}L\right),
\end{align}
where $c_{L,V,T}$ defined in the following has the property that for every $T\in\bbN$,  $c_{L,V,T}$ tends to 1 as $\frac{L}{\sqrt{V}}$ tends to infinity:
\begin{equation}
c_{L,V,T} :=  \frac{1+(L/\sqrt{V})^2 T}{(L/\sqrt{V})^2 T} \cdot \frac{1}{1-\exp\{-(L/\sqrt{V})^2 /2\}}. \label{eqn:def_C}
\end{equation}
\end{corollary}
\begin{figure*}[t]
 \centering
  {
  \includegraphics[width=150mm]{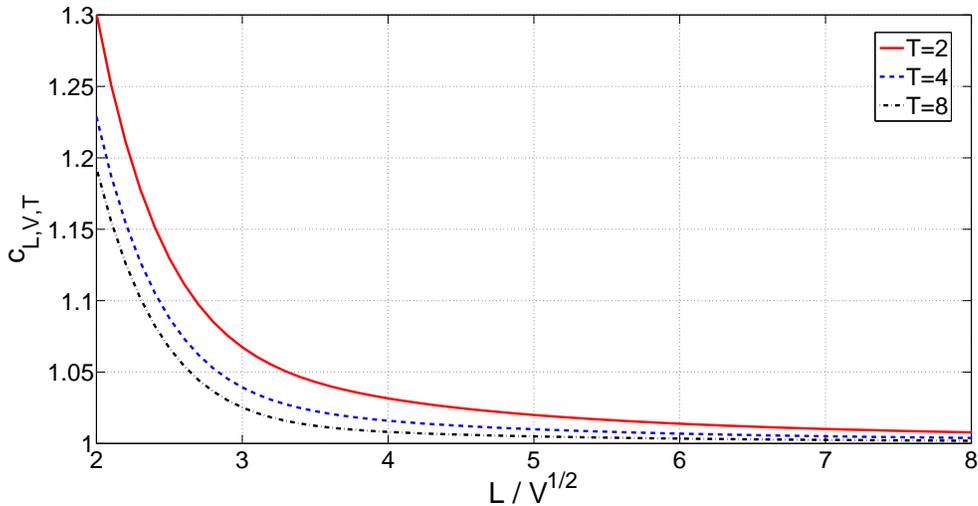}}
  \caption{The constant $c_{L,V,T}$ in Theorem \ref{thm:CL} is illustrated in terms of $\frac{L}{\sqrt{V}}$. } \label{fig:const}
\end{figure*}
Fig. \ref{fig:const} illustrates how fast the constant $c_{L,V,T}$ in Corollary \ref{coro:CL_asymp} converges to 1 as  $\frac{L}{\sqrt{V}}$ increases. For $T=2$, we can see that $c_{L,V,T}$ is less than 1.1 when $\frac{L}{\sqrt{V}}=3$ and is less than 1.05 when $\frac{L}{\sqrt{V}}=4$. Hence, the effect of the constant $c_{L,V,T}$ is not significant for a wide range of $L, V,$ and $T$.

Theorems \ref{thm:MD} and \ref{thm:CL} illustrate that the fundamental interplay between the rate and probability of  error can be greatly improved when delay is allowed in the streaming setup. In the moderate deviations regime, the moderate deviations constant improves by a factor of $T$. Assuming that $c_{L,V,T}$ can be approximated sufficiently well by $1$, for the central limit regime, the second-order coding rate $L$ is improved (reduced) by a factor of $\sqrt{T}$. Another way to view this via the lens of the channel dispersion $V$; this parameter is approximately reduced by a factor of $T$.

\section{Proofs of the Main Theorems} \label{sec:proof}
\subsection{Proof of Theorem \ref{thm:MD} for the moderate deviations regime}  \label{subsec:MD_pf}
Consider a DMC $(\mathcal{X},\mathcal{Y}, \{W(y|x): x\in \mathcal{X}, y\in \mathcal{Y}\} )$ with $V>0$ and any sequence of integers $M_n$ such that $\log M_n=nC-n\rho_n$, where $\rho_n>0, \rho_n\rightarrow 0$ and $n\rho_n^2\rightarrow \infty$. We denote by $P_X$ an input distribution that achieves the dispersion \eqref{eqn:dispersion}.

\subsubsection{Encoding} \label{subsubsec:MD_encoding}
For each $k\in \bbN$ and $g^k \in \mathcal{G}^{k}$, generate $\mathbf{x}_k(g^k)$ in an i.i.d. manner according to $P_X$. The generated codewords constitute the codebook $\mathcal{C}_n$. In block $k$, after observing the true message sequence $G^k$, the encoder sends $\mathbf{x}_k(G^k)$. 

\subsubsection{Decoding} 
Consider the decoding of $G_k$ at the end of block $T_k:=k+T-1$. In our scheme, the decoder not only decodes $G_k$, but also re-decodes $G_1,\cdots, G_{k-1}$ at the end of block $T_k$.\footnote{We note that $G_j$ for $j\in [1:k-1]$ has been already decoded at  the end of block $T_j$. Nevertheless, the decoder re-decodes $G^{k-1}$ at the end of $T_k$, because the decoder needs to decode $G^{k-1}$ to decode $G_k$ and the probability of error associated with $G^{k-1}$ becomes lower (in general) by utilizing recent channel output sequences.} Let $\hat{G}_{T_k,j}$ denote the estimate of $G_j$ at the end of block $T_k$. The decoder decodes $G_j$ sequentially from $j=1$ to $j=k$ as follows: 
\begin{itemize}
\item Given $\hat{G}_{T_k,[1:j-1]}$, the decoder chooses $\hat{G}_{T_k,j}$ according to the following rule.\footnote{When $j=1$, $\hat{G}^{j-1}_{T_k}$ is null. } If there is a unique index $g_j\in \mathcal{G}$ that satisfies\footnote{We use the following notation for the set of codewords. Let $\mathcal{K}_j$ for $j\in \bbN$ denote the set of message indices mapped to the $j$-th codeword according to the encoding procedure. For $\mathcal{J}\subseteq \bbN$ and $\mathcal{K}\supseteq \bigcup_{j\in \mathcal{J}} \mathcal{K}_j$, we denote by  $\mathbf{x}_{\mathcal{J}}(g_{\mathcal{K}})$ the set of codewords $\{\mathbf{x}_j(g_{\mathcal{K}_j}): j\in \mathcal{J}\}$. } 
\begin{align}
i(\mathbf{x}_{[j:T_k]}(\hat{G}_{T_k, [1:j-1]}, g_{[j:T_k]}), \mathbf{y}_{[j:T_k]})&> (T_k-j+1) \cdot \log M_n \label{eqn:dec_rule}
\end{align}
for some $g_{[j+1:T_k]}$, let $\hat{G}_{T_k,j}=g_j$.\footnote{We note that $i(\cdot, \cdot)$ in \eqref{eqn:dec_rule} is defined in terms of $P_X$ and $W$. This dependence is suppressed henceforth.} 
If there is none or more than one such $g_j$, let $\hat{G}_{T_k,j}=1$. 
\item If $j<k$, repeat the above procedure by increasing $j$ to $j+1$. If $j=k$, the decoding procedure terminates and the decoder declares that the $k$-th message is $\hat{G}_{k}:=\hat{G}_{T_k,k}$.
\end{itemize}

\subsubsection{Error analysis} 
We first consider the probability of error averaged over random codebook $\mathcal{C}_n$.  The error event $\{\hat{G}_k\neq G_k\}$ for $k\in \bbN$ happens only if at least one of the following $2k$ events occurs: 
\begin{align}
\mathcal{E}_{k,j}&:= \{
i(\mathbf{X}_{[j:T_k]}(G^{T_k}), \mathbf{Y}_{[j:T_k]})\leq (T_k-j+1) \cdot \log M_n
\},~ j\in [1:k] \label{eqn:err1}\\
\tilde{\mathcal{E}}_{k,j}&:= \{i(\mathbf{X}_{[j:T_k]}(G^{j-1},g_{[j:T_k]}), \mathbf{Y}_{[j:T_k]})> (T_k-j+1) \cdot \log M_n \cr
&~~~~~~~~~~~~~~~~~~~~~~~~~~~~~~~~~~~~ \mbox{ for some }g_{[j:T_k]} \mbox{ such that }g_{j} \neq G_{j} \},~ j\in [1:k]. \label{eqn:err2}
\end{align}
Now, we have
\begin{align}
\E_{\mathcal{C}_n}[\Pr(\hat{G}_k\neq G_k|\mathcal{C}_n)]&\leq \sum_{j=1}^k \left(\Pr(\mathcal{E}_{k,j})+\Pr(\tilde{\mathcal{E}}_{k,j}) \right).
\end{align}
For each $j\in [1:k]$, we have 
\begin{align}
\Pr(\mathcal{E}_{k,j})+\Pr(\tilde{\mathcal{E}}_{k,j})&\leq \Pr\left(\sum_{l=1}^{n (T_k-j+1) } i(X_l;Y_l)\leq (T_k-j+1) \cdot \log M_n\right)\cr
&~~~~+M_n^{T_k-j+1} \Pr\left(\sum_{l=1}^{n (T_k-j+1) } i(X_l;\bar{Y}_l)>  (T_k-j+1)  \log M_n\right)\\
&\overset{(a)}{=}\E\left[\exp\left\{-\left[\sum_{l=1}^{n(T_k-j+1)} i(X_l;Y_l)- (T_k-j+1) \log M_n\right]^+\right\}\right]\\
&=\E\left[\exp\left\{-\left[\sum_{l=1}^{n(T_k-j+1)} i(X_l;Y_l)- (T_k-j+1)n(C-\rho_n)\right]^+\right\}\right], \label{eqn:fixj}
\end{align}
where $(X_l,Y_l, \bar{Y}_l)$'s are i.i.d. random variables each generated according to $P_X(x_l)W(y_l|x_l)P_XW(\bar{y}_l)$ and $(a)$ is from the identity \cite[Eq.~(69)]{PolyanskiyPoorVerdu:10} used to derive the DT bound.

Now, fix an arbitrary $0<\lambda<1$.  By applying the chain of inequalities \cite[Eq. (53)-(56)]{PolyanskiyVerdu:10}, we have 
\begin{align}
&\exp\left\{-\left[\sum_{l=1}^{n(T_k-j+1)} i(X_l;Y_l)- (T_k-j+1)n(C-\rho_n)\right]^+\right\}\cr
&\leq  \mathbbm{1}\left\{\sum_{l=1}^{n(T_k-j+1)} i(X_l;Y_l)\leq (T_k-j+1)n(C-\lambda\rho_n)\right\} +\exp\left\{-(T_k-j+1)n(1-\lambda)\rho_n\right\}.  \label{eqn:spr}
\end{align}
Combining the bounds in \eqref{eqn:fixj} and \eqref{eqn:spr}, we obtain
\begin{align}
&\Pr(\mathcal{E}_{k,j})+\Pr(\tilde{\mathcal{E}}_{k,j})\cr
&\leq \Pr\left(\sum_{l=1}^{n(T_k-j+1)} i(X_l;Y_l)\leq (T_k-j+1)n(C-\lambda\rho_n)\right)+\exp\left\{-(T_k-j+1)n(1-\lambda)\rho_n\right\}\label{eqn:spr2}\\
&\overset{(a)}{\leq} \exp\left\{-(T_k-j+1)n\left(\frac{\lambda^2 \rho_n^2}{2V}-\lambda^3 \rho_n^3 \tau\right)\right\}+\exp\left\{-(T_k-j+1)n(1-\lambda)\rho_n\right\} 
\end{align}
for sufficiently large $n$, where $\tau$ is some non-negative constant dependent only on the input distribution $P_X(x)$ and channel statistics $W(y|x)$ and $(a)$ is from the moderate deviations upper bound in Lemma~\ref{lemma:nonasymptotic_MD}, which is relegated to the end of this subsection. Also see Remark \ref{rem:MD}.

Now, we have 
\begin{align}
&\E_{\mathcal{C}_n}[\Pr(\hat{G}_k\neq G_k|\mathcal{C}_n)]\cr
&\leq \sum_{j=1}^k \left( \exp\left\{-(T_k-j+1)n\rho_n^2\lambda^2\left(\frac{1}{2V}-\lambda \rho_n \tau\right)\right\}+\exp\left\{-(T_k-j+1)n(1-\lambda)\rho_n\right\} \right)\\
&\leq \sum_{j=T}^{T_k} \left( \exp\left\{-jn\rho_n^2\lambda^2\left(\frac{1 }{2V}-\lambda \rho_n \tau\right)\right\}+\exp\left\{-jn(1-\lambda)\rho_n\right\}\right) \\
&\leq \frac{\exp\left\{-Tn\rho_n^2\lambda^2\left(\frac{1 }{2V}-\lambda \rho_n \tau\right)\right\}}{1-\exp\{-n\rho_n^2\lambda^2\left(\frac{1 }{2V}-\lambda \rho_n \tau\right)\}}+\frac{\exp\left\{-Tn(1-\lambda)\rho_n\right\}}{1-\exp\left\{-n(1-\lambda)\rho_n\right\}} \label{eqn:MD_sumup}
\end{align}
for sufficiently large $n$, which leads to 
\begin{align}
\limsup_{n\rightarrow \infty} \frac{1}{n\rho_n^2 }\log \E_{\mathcal{C}_n}\left[ \limsup_{N\rightarrow \infty} \frac{\sum_{k=1}^N \Pr(\hat{G}_k\neq G_k|\mathcal{C}_n)}{N}\right]\leq -\frac{T\lambda^2}{2V}.
\end{align}
Finally, by taking $\lambda\rightarrow 1$, we have
\begin{align}
\limsup_{n\rightarrow \infty} \frac{1}{n\rho_n^2 }\log \E_{\mathcal{C}_n}\left[ \limsup_{N\rightarrow \infty} \frac{\sum_{k=1}^N \Pr(\hat{G}_k\neq G_k|\mathcal{C}_n)}{N}\right]\leq -\frac{T}{2V}.
\end{align}
Hence, there must exist a sequence of codes $\mathcal{C}_n$ that satisfies  \eqref{eqn:MD_positiveV}, which completes the proof. \endproof

The following lemma used in the proof of Theorem \ref{thm:MD} corresponds to a non-asymptotic upper bound of the moderate deviations theorem \cite[Theorem 3.7.1]{DemboZeitouni:09}, whose proof is in Appendix \ref{appendix:finite_MD}.
\begin{lemma} \label{lemma:nonasymptotic_MD}
Let $\{Z_l\}_{l\geq 1}$ be a sequence of i.i.d.  random variables such that 
$\E[Z_1] = 0$, $\var[Z_1] = \sigma^2>0$, and its cumulant generating function $h(s) := \log\E[\exp\{sZ_1\}]$ for $s\geq 0$ is  analytic around the origin and satisfies that $K:=\max_{s \in [0,1] }  | h'''(s) |$ is finite. 
For a sequence $\eps_n>0$ satisfying  the moderate deviations constraints, i.e.,  $\eps_n\to 0$ and $n\eps_n^2\to\infty$, the following bound holds: 
\begin{align}
\Pr\left(\frac{1}{n}\sum_{l=1}^n Z_l \geq \eps_n\right) &\leq \exp\left\{ -n \left(  \frac{\eps_n^2}{2\sigma^2}  -\frac{\eps_n^3}{6\sigma^6}  K \right)\right\}
\end{align}
for sufficiently large $n$. 
\end{lemma}
\begin{remark} \label{rem:MD}
Let us comment on the assumption in Lemma \ref{lemma:nonasymptotic_MD} that $K$ is finite. In our application, 
\begin{equation}
Z_l \equiv  i(X_l;Y_l) - I(X_l;Y_l).
\end{equation}
Then, we have
 \begin{align}
 h (s) &= \log \E \left[ \exp \bigg\{ s\Big( \log\frac{W(Y_1 |X_1)}{P_XW(Y_1)} - I(X_1;Y_1) \Big) \bigg\} \right]\\
  &= -s I(X_1;Y_1) + \log \E \left[  \Big(\frac{W(Y_1 |X_1)}{P_XW(Y_1)}   \Big)^s\right].
 \end{align}
By differentiating thrice, we can show that $h'''(s)$ is continuous in $s$.\footnote{A detailed calculation follows similarly as in the proof of \cite[Lemma 1]{AltugWagner:14}.} Restricting $s$ to $[0,1]$ means that $h'''(s)$ is a continuous function over a compact set. Hence its maximum is attained and is necessarily finite. 
\end{remark}

\subsection{Proof of Theorem \ref{thm:CL} for the central limit regime}
Consider a DMC $(\mathcal{X},\mathcal{Y}, \{W(y|x): x\in \mathcal{X}, y\in \mathcal{Y}\} )$ with $V>0$. We remark that in the moderate deviations regime,  for every block, the encoder maps \emph{all} the previous messages to a codeword. For the central limit regime, we propose a coding strategy where the encoder maps only \emph{some} recent  messages to the codeword in each block. 
Similar idea of incorporating truncated memory was used in \cite{DraperKhisti:11} with the focus on reducing the complexity. Here, we use a different memory structure from \cite{DraperKhisti:11}. 
Let $A\in \bbN$ and $B\in \bbN$ denote the maximum and the minimum numbers of messages that can possibly be mapped to a codeword in each block, respectively. We choose the size  $M_n$ of message alphabet as follows: 
\begin{align}
\log M_n=\frac{A-2B+T+2}{A}(nC-L\sqrt{n}) \label{eqn:CL_msg}
\end{align}
for some $L>0$. To make the above choice of $M_n$ valid, we assume $A\geq 2B-T-2\geq 0$. Furthermore, we assume that the minimum encoding memory is at least $T$, i.e., $B\geq T$. We denote by $P_X$ an input distribution that achieves the dispersion \eqref{eqn:dispersion}.

\subsubsection{Encoding}
Our encoder has a periodically time-varying memory $m \in [B:A]$ with a period of $A-B+1$ blocks, after an initialization step of the first $A$ blocks. 
Let us first describe our message-codeword mapping rule for the case of $A=9$ and $B=4$, which is illustrated in Fig. \ref{fig:mapping}. For the first nine blocks, the encoder maps all the previous messages to a codeword. Since the maximum encoding memory is nine in this example, we \emph{truncate} the messages that are mapped to a codeword on and after the tenth block, so that the encoding memory is periodically time-varying from four to nine with a period of six blocks.  
For instance, let us consider the first period  from the tenth block to the fifteenth block. In the tenth block, the encoder maps the messages $G_{7},\cdots, G_{10}$ to a codeword, thus ensuring that the encoding memory is four. In block $k\in [11: 15]$, the encoder maps the messages $G_{7},\cdots, G_{k}$ to a codeword and hence the encoding memory becomes the maximum memory of nine when $k=15$.
\begin{figure*}[t]
 \centering
  {
  \includegraphics[width=110mm]{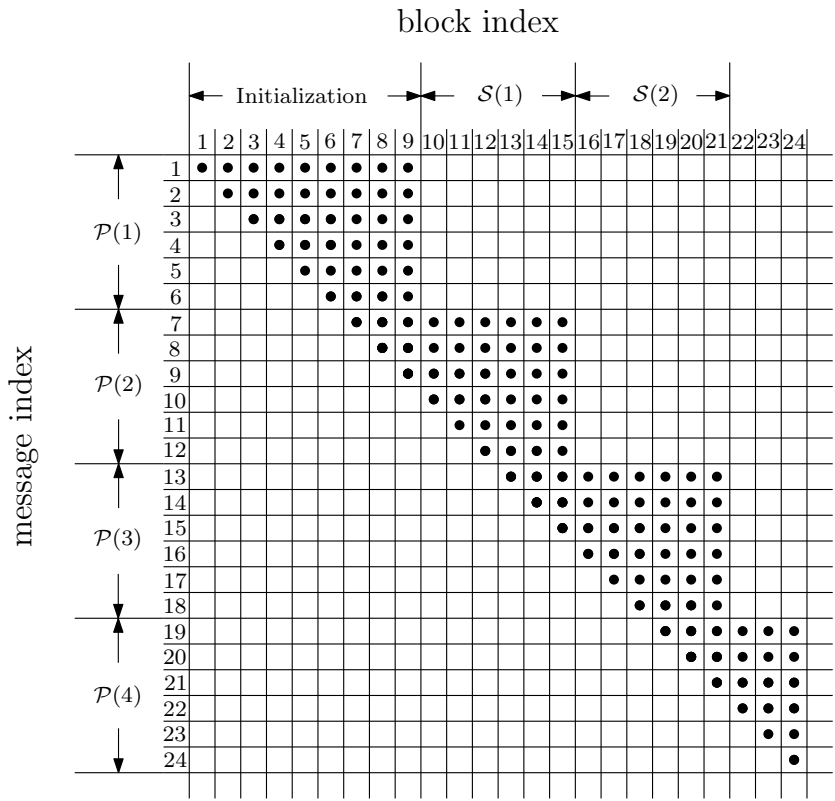}}
  \caption{The proposed message-codeword mapping rule for the central limit regime is illustrated for the case of $A=9$ (maximum encoding memory) and $B=4$ (minimum encoding memory). After an initialization step of the first nine blocks, in which all the previous messages are mapped to a codeword, our encoder has a periodically time-varying memory from four to nine with a period of six blocks. 
} \label{fig:mapping}
\end{figure*}

Now, let us formally describe the encoding procedure for the general case.  
For each $k\in [1:A]$ and $g^k \in \mathcal{G}^{k}$, generate $\mathbf{x}_k(g^k)$ in an i.i.d. manner according to $P_X$. In block $k\in [1:A]$, the encoder sends $\mathbf{x}_k(G^k)$.
Let $\mathcal{S}(q)$ for $q\geq 1$ denote the set of $(A-B+1)$ block indices in the $q$-th period on and after the $(A+1)$-st block, i.e., $\mathcal{S}(q)=\{(A-B+1)q+B, \cdots, (A-B+1)(q+1)+B-1\}$. For each $k\in \mathcal{S}(q)$ and $g^{k-q(A-B+1)} \in \mathcal{G}^{k-q(A-B+1)}$,\footnote{In block $k\in \mathcal{S}(q)$, a total of $k-q(A-B+1)$ messages, i.e., $G_{q(A-B+1)+1},\cdots, G_k$, are mapped to a codeword.} generate $\mathbf{x}_k(g^{k-q(A-B+1)})$ in an i.i.d. manner according to $P_X$. In block $k\in  \mathcal{S}(q)$, the encoder sends $\mathbf{x}_k(G_{[q(A-B+1)+1:k]})$. 

On the other hand, we note that our message-codeword mapping rule is also periodic in the (vertical) axis of message index. We can group the messages according to the maximum block index to which a message is mapped. Let $\mathcal{P}(q)$ for $q\in \bbN$ denote the $q$-th group $\{G_{(A-B+1)(q-1)+1},\cdots,G_{(A-B+1)q}\}$ of messages that are mapped to a codeword up to block $(A-B+1)q+B-1$, which is illustrated in Fig. \ref{fig:mapping} for the example of $A=9$ and $B=4$.  This grouping rule is useful for describing the decoding rule.

\subsubsection{Decoding} 
The decoding rule of $G_k\in \mathcal{P}(1)$ at the end of block $T_k$ is exactly the same as that for the moderate deviations regime. Hence, from now on, let us focus on the decoding of $G_k\in \mathcal{P}(q)$ for $q\geq 2$ at the end of block $T_k$.  At the end of block $T_k$, the decoder decodes not only $G_k$, but also all the messages in the previous group and the previous messages in the current group,\footnote{Similarly as in the moderate deviations regime, $G_j$ for $j\in [1:k-1]$ has been already decoded at  the end of block $T_j$. Nevertheless, the decoder re-decodes some of the previous messages at the end of $T_k$.} i.e., $G_{(A-B+1)(q-2)+1}, \cdots, G_{k-1}$. Let $\hat{G}_{T_k,j}$ denote the estimate of $G_j$ at the end of block $T_k$.

Let us first describe our decoding procedure for the example of $T=2$, $A=9$, and $B=4$ illustrated in Fig.~\ref{fig:mapping}. Consider the decoding of $G_k\in \mathcal{P}(3)=\{G_{13},\cdots, G_{18}\}$ at the end of block $T_k$.\footnote{By using the symmetry of the message-codeword mapping rule, the procedure for decoding $G_k\in \mathcal{P}(q)$ for the cases $q=2$ and $q\geq 4$ can be stated in a similar manner.} The decoder decodes not only $G_k$, but also all the messages $G_7,\cdots, G_{12}$ in $\mathcal{P}(2)$ and the previous messages $G_{13},\cdots,G_{k-1}$ in $\mathcal{P}(3)$. The underlying rules of our decoding procedure can be summarized as follows: 
\begin{itemize}
\item Since messages $G_1,\cdots,G_6$ in $\mathcal{P}(1)$, which we do not want to decode, are involved in blocks $1, \cdots, 9$, we do not utilize the channel output sequences in those blocks for decoding.  
\item For the decoding of the $j$-th message for $j\in [7:k]$, among the channel output sequences from block $10$ to block $T_k$, we utilize the channel output sequences in which the $j$-th message is involved.
\end{itemize}
According to the above rules, the blocks to be considered for the decoding of messages $G_7,\cdots, G_k$ are as follows: 
\begin{enumerate}[(i)]
\item for $G_7, \cdots, G_{10}$, blocks\footnote{We note the last block index to which the messages in $\mathcal{P}(2)$ are involved is $T_k$ if $T_k\leq 15$, and it is $15$ otherwise. In other words, the last block index to which the messages in $\mathcal{P}(2)$ are involved is $\min(T_k, 15)$.} indexed from 10 to $\nu_k:=\min(T_k, 15)$, 
\item for $G_j$ for $j\in [11:12]$, blocks indexed from $j$ to $\nu_k$, and 
\item for $G_{j}$ for $j\in [13:k]$, blocks indexed from $j$ to $T_k$.
\end{enumerate}

 In particular, since the pairs of the first block index and the last block index to be considered for the decoding of messages $G_7, \cdots, G_{10}$ are the same, we decode   $G_7, \cdots, G_{10}$ simultaneously. By keeping this in mind, our decoding procedure for $G_k\in \mathcal{P}(3)$ for the example of $T=2$, $A=9$ and $B=4$ is formally stated as follows: 
\begin{enumerate}[(i)]
\item If there is a unique index vector $g_{[7:10]}$ that satisfies\footnote{Similarly as in the proof of Theorem \ref{thm:CL}, the following notation is used for the set of codewords. Let $\mathcal{K}_j$ for $j\in \bbN$ denote the set of message indices mapped to the $j$-th codeword according to the encoding procedure. For $\mathcal{J}\subseteq \bbN$ and $\mathcal{K}\supseteq \bigcup_{j\in \mathcal{J}} \mathcal{K}_j$, we denote by  $\mathbf{x}_{\mathcal{J}}(g_{\mathcal{K}})$ the set of codewords $\{\mathbf{x}_j(g_{\mathcal{K}_j}): j\in \mathcal{J}\}$. } 
\begin{align}
i(\mathbf{x}_{[10:\nu_k]}(g_{[7:\nu_k]}),\mathbf{y}_{[10:\nu_k]})&>(\nu_k-6)\cdot \log M_n\label{eqn:dec_rule_CL1_toy}
\end{align}
for some $g_{[11:\nu_k]}$, let $\hat{G}_{T_k,[7:10]}=g_{[7:10]}$. If there is none or more than one such $g_{[7:10]}$, let $\hat{G}_{T_k,[7:10]}=(1,\cdots, 1)$.

\item  The decoder sequentially decodes $g_j$ from $j=11$ to $j=12$ as follows: 
\begin{itemize}
\item Given $\hat{G}_{T_k,[7:j-1]}$, the decoder chooses $\hat{G}_{T_k,j}$ according to the following rule. If there is a unique index $g_j\in \mathcal{G}$ that satisfies 
\begin{align}
i(\mathbf{x}_{[j:\nu_k]}(\hat{G}_{T_k,[7:j-1]}, g_{[j:\nu_k]}), \mathbf{y}_{[j:\nu_k)]})&> (\nu_k-j+1) \cdot \log M_n \label{eqn:dec_rule_CL2_toy}
\end{align}
for some $g_{[j+1:\nu_k]}$, let $\hat{G}_{T_k,j}=g_j$. If there is none or more than one such $g_j$, let $\hat{G}_{T_k,j}=1$. 
\item If $j=11$, repeat the above procedure by increasing $j$ to $12$. If $j=12$, proceed to the next decoding procedure.
\end{itemize}

\item  The decoder sequentially decodes $g_j$ from $j=13$ to $j=k$ as follows: 
\begin{itemize}
\item Given $\hat{G}_{T_k, [7:j-1]}$, the decoder chooses $\hat{G}_{T_k, j}$ according to the following rule. If there is a unique index $g_j\in \mathcal{G}$ that satisfies
\begin{align}
i(\mathbf{x}_{[j:T_k]}(\hat{G}_{T_k,[7:j-1]}, g_{[j:T_k]}), \mathbf{y}_{[j:T_k]})&> (T_k-j+1) \cdot \log M_n \label{eqn:dec_rule_CL3_toy}
\end{align}
for some $g_{[j+1:T_k]}$, let $\hat{G}_{T_k,j}=g_j$. If there is none or more than one such $g_j$, let $\hat{G}_{T_k,j}=1$. 
\item If $j<k$, repeat the above procedure by increasing $j$ to $j+1$. If $j=k$, the whole decoding procedure terminates and the decoder declares that the $k$-th message is $\hat{G}_{k}:=\hat{G}_{T_k,k}$.
\end{itemize}
\end{enumerate}

The above description of the decoding procedure for the example in Fig. \ref{fig:mapping} is  naturally extended for the general case. 
In general, the procedure for decoding of $G_k\in \mathcal{P}(q)$ for $q\geq 2$ at the end of block $T_k$  consists of the following three steps: (i) simultaneous non-unique decoding of the first $B$ messages in the previous group, (ii) sequential decoding of the remaining $A-2B+1$ messages in the previous group, and (iii) sequential decoding of the messages in the current group up to the current block.  
Let us describe the decoding rule when $q=2$ in the following: 
\begin{enumerate}[(i)]
\item If there is a unique index vector $g^B$ that satisfies 
\begin{align}
i(\mathbf{x}_{[B:\min(A,T_k) ]}(g^{\min(A,T_k)}),\mathbf{y}_{[B:\min (A,T_k)]})&>\min(A,T_k)\cdot \log M_n\label{eqn:dec_rule_CL1}
\end{align}
for some $g_{[B+1:\min(A,T_k)]}$, let $\hat{G}_{T_k, [1:B]}=g^B$. If there is none or more than one such $g^B$, let $\hat{G}_{T_k,[1:B]}=(1,\cdots,1)$.

\item  The decoder sequentially decodes $g_j$ from $j=B+1$ to $j=A-B+1$ as follows: 
\begin{itemize}
\item Given $\hat{G}_{T_k,[1:j-1]}$, the decoder chooses $\hat{G}_{T_k,j}$ according to the following rule. If there is a unique index $g_j\in \mathcal{G}$ that satisfies 
\begin{align}
i(\mathbf{x}_{[j:\min(A,T_k)]}(\hat{G}_{T_k,[1:j-1]}, g_{[j:\min(A,T_k)]}), \mathbf{y}_{[j:\min(A,T_k)]})&> (\min(A,T_k)-j+1) \cdot \log M_n \label{eqn:dec_rule_CL2}
\end{align}
for some $g_{[j+1:\min(A,T_k)]}$, let $\hat{G}_{T_k,j}=g_j$. If there is none or more than one such $g_j$, let $\hat{G}_{T_k,j}=1$. 
\item If $j<A-B+1$, repeat the above procedure by increasing $j$ to $j+1$. If $j=A-B+1$, proceed to the next decoding procedure.
\end{itemize}

\item  The decoder sequentially decodes $g_j$ from $j=A-B+2$ to $j=k$ as follows: 
\begin{itemize}
\item Given $\hat{G}_{T_k,[1:j-1]}$, the decoder chooses $\hat{G}_{T_k,j}$ according to the following rule. If there is a unique index $g_j\in \mathcal{G}$ that satisfies 
\begin{align}
i(\mathbf{x}_{[j:T_k]}(\hat{G}_{T_k,[1:j-1]}, g_{[j:T_k]}), \mathbf{y}_{[j:T_k]})&> (T_k-j+1) \cdot \log M_n \label{eqn:dec_rule_CL3}
\end{align}
for some $g_{[j+1:T_k]}$, let $\hat{G}_{T_k,j}=g_j$. If there is none or more than one such $g_j$, let $\hat{G}_{T_k,j}=1$. 
\item If $j<k$, repeat the above procedure by increasing $j$ to $j+1$. If $j=k$, the whole decoding procedure terminates and the decoder declares that the $k$-th message is $\hat{G}_{k}:=\hat{G}_{T_k,k}$.
\end{itemize}
\end{enumerate}
By exploiting the symmetry of the message-codeword mapping rule, the decoding rule for $q\geq 3$ proceeds similarly.

\subsubsection{Error analysis}
We first consider the probability of error averaged over random codebook $\mathcal{C}_n$. Let us consider the decoding of $G_k\in \mathcal{P}(2)$.  Let $\alpha:=\min(A,T_k)$. The error event $\{\hat{G}_k\neq G_k\}$ happens only if at least one of the following $2(k-B+1)$ events occurs: 
\begin{align}
\mathcal{E}_{k}^{\mathrm{(i)}}&:= 
\{i(\mathbf{X}_{[B:\alpha]}(G^{\alpha}), \mathbf{Y}_{[B:\alpha]})\leq \alpha\cdot \log M_n\} \\
\tilde{\mathcal{E}}_{k}^{\mathrm{(i)}}&:= 
\{i(\mathbf{X}_{[B:\alpha]}(g^{\alpha}), \mathbf{Y}_{[B:\alpha]})>  \alpha \cdot \log M_n \mbox{ for some } g^{\alpha} \mbox{ such that }g^{B} \neq G^B \} \\
\mathcal{E}_{k,j}^{\mathrm{(ii)}}&:= \{i(\mathbf{X}_{[j:\alpha]}(G^{\alpha}), \mathbf{Y}_{[j:\alpha]})\leq (\alpha-j+1) \cdot \log M_n\} \mbox{ for } j\in [B+1:A-B+1] \\
\tilde{\mathcal{E}}_{k,j}^{\mathrm{(ii)}}&:= \{i(\mathbf{X}_{[j:\alpha]}(G^{j-1},g_{[j:\alpha]}), \mathbf{Y}_{[j:\alpha]})> (\alpha-j+1) \cdot \log M_n \cr
& ~~~~~~~~~~~~~~~~~~ \mbox{ for some }g_{[j:\alpha]} \mbox{ such that }g_{j} \neq G_{j} \} \mbox{ for } j\in [B+1:A-B+1]\\
\mathcal{E}_{k,j}^{\mathrm{(iii)}}&:=  \{i(\mathbf{X}_{[j:T_k]}(G^{T_k}), \mathbf{Y}_{[j:T_k]})\leq (T_k-j+1) \cdot \log M_n\} \mbox{ for } j\in [A-B+2:k]\\
\tilde{\mathcal{E}}_{k,j}^{\mathrm{(iii)}}&:=\{i(\mathbf{X}_{[j:T_k]}(G^{j-1},g_{[j:T_k]}), \mathbf{Y}_{[j:T_k]})> (T_k-j+1) \cdot \log M_n \cr
&  ~~~~~~~~~~~~~~~~~~\mbox{ for some }g_{[j:T_k]} \mbox{ such that }g_{j} \neq G_{j} \} \mbox{ for }j\in [A-B+2:k].
\end{align}

We note that the superscript in each error event represents the decoding step in which the error event is involved. Now, we have 
\begin{align}
\E_{\mathcal{C}_n}[\Pr(\hat{G}_k\neq G_k|\mathcal{C}_n)]
&\leq \Pr(\mathcal{E}_{k}^{\mathrm{(i)}})+\Pr( \tilde{\mathcal{E}}_{k}^{\mathrm{(i)}})+ \sum_{j=B+1}^{A-B+1} \Pr(\mathcal{E}_{k,j}^{\mathrm{(ii)}}) 
+\sum_{j=B+1}^{A-B+1} \Pr( \tilde{\mathcal{E}}_{k,j}^{\mathrm{(ii)}}) \cr
&\quad \quad \quad\quad \quad+\sum_{j=A-B+2}^{k} \Pr( \mathcal{E}_{k,j}^{\mathrm{(iii)}}) 
+\sum_{j=A-B+2}^{k} \Pr( \tilde{\mathcal{E}}_{k,j}^{\mathrm{(iii)}}). \label{eqn:first_case_sum}
\end{align}
Let us bound each term in the RHS of (\ref{eqn:first_case_sum}). First, $\Pr(\mathcal{E}_{k}^{\mathrm{(i)}})$ is upper-bounded as follows: 
\begin{align}
\Pr(\mathcal{E}_{k}^{\mathrm{(i)}})&=\Pr\left(i(\mathbf{X}_{[B:\alpha]}(G^{\alpha}), \mathbf{Y}_{[B:\alpha]})\leq \alpha\cdot \log M_n\right)\\
&\leq \Pr\left(\sum_{l=1}^{n(\alpha-B+1)} i(X_l;Y_l)\leq \alpha \cdot \log M_n\right)\\
&\overset{(a)}{\leq} \Pr\left(\sum_{l=1}^{n(\alpha-B+1)} i(X_l;Y_l)\leq (\alpha-B+1)(nC-L\sqrt{n})\right)\\
&\overset{(b)}{\leq} Q\left(\frac{\sqrt{\alpha-B+1}}{\sqrt{V}}L\right)+\frac{\tau_1}{\sqrt{(\alpha-B+1)n}} \label{eqn:berry_1}
\end{align} 
for some non-negative constant $\tau_1$ that is dependent only on the input distribution $P_X$ and the channel statistics $W(y|x)$, where $(X_l,Y_l)$'s are i.i.d. random variables each generated according to $P_X(x_l)W(y_l|x_l)$, $(a)$ is from the choice of $M_n$ in \eqref{eqn:CL_msg}, and $(b)$ is from the Berry-Esseen Theorem (e.g., \cite{Feller:71}). Similarly, we can show 
\begin{align}
\sum_{j=B+1}^{A-B+1} \Pr(\mathcal{E}_{k,j}^{\mathrm{(ii)}}) &\leq \sum_{j=B+1}^{A-B+1} Q\left(\frac{\sqrt{\alpha-j+1}}{\sqrt{V}}L\right)+\frac{\tau_1}{\sqrt{(\alpha-j+1)n}} \label{eqn:berry_2}
\end{align} 
and
\begin{align}
\sum_{j=A-B+2}^{k} \Pr( \mathcal{E}_{k,j}^{\mathrm{(iii)}}) &\leq\sum_{j=A-B+2}^{k} Q\left(\frac{\sqrt{T_k-j+1}}{\sqrt{V}}L\right)+\frac{\tau_1}{\sqrt{(T_k-j+1)n}}.\label{eqn:berry_3}
\end{align}

Next, $\Pr( \tilde{\mathcal{E}}_{k}^{\mathrm{(i)}})$ is upper-bounded as follows: 
\begin{align}
\Pr( \tilde{\mathcal{E}}_{k}^{\mathrm{(i)}})
&=\Pr(i(\mathbf{X}_{[B:\alpha]}(g^{\alpha}), \mathbf{Y}_{[B:\alpha]})>  \alpha \cdot \log M_n \mbox{ for some } g^{\alpha} \mbox{ such that }g^{B} \neq G^B)\\
&\leq  M_n^{\alpha} \cdot \Pr\left(\sum_{l=1}^{n(\alpha-B+1)} i(X_l;\bar{Y}_l)> \alpha\cdot \log M_n\right)\\
&\overset{(a)}{=} M_n^{\alpha}\cdot \E\left[\exp\left\{-\sum_{l=1}^{n(\alpha-B+1)} i(X_l;Y_l)\right\}\right. \cdot \left.\mathbbm{1}\Big\{\sum_{l=1}^{n(\alpha-B+1)} i(X_l;Y_l)>\alpha \log M_n\Big\}\right]\\
&\overset{(b)}{\leq} \frac{\tau_2}{\sqrt{(\alpha-B+1)n}}
\end{align}
for  some non-negative constant $\tau_2$  that is dependent only on the input distribution $P_X$ and channel statistics $W(y|x)$, where $(X_l,Y_l, \bar{Y}_l)$'s are i.i.d. random variables each generated according to $P_X(x_l)W(y_l|x_l)$ $P_XW(\bar{y}_l)$, $(a)$ is due to an elementary chain of  equalities given in Appendix \ref{appendix:chain}, and  $(b)$ is from \cite[Lemma 47]{PolyanskiyPoorVerdu:10}.

Similarly, we can show
\begin{align}
\sum_{j=B+1}^{A-B+1} \Pr( \tilde{\mathcal{E}}_{k,j}^{\mathrm{(ii)}})\leq \sum_{j=B+1}^{A-B+1} \frac{\tau_2}{\sqrt{(\alpha-j+1)n}}
\end{align}
and 
\begin{align}
\sum_{j=A-B+2}^{k} \Pr( \tilde{\mathcal{E}}_{k,j}^{\mathrm{(iii)}}) \leq \sum_{j=A-B+2}^{k} \frac{\tau_2}{\sqrt{(T_k-j+1)n}} .
\end{align}

By substituting the above bounds into the RHS of \eqref{eqn:first_case_sum}, we obtain
\begin{align}
&\E_{\mathcal{C}_n}[\Pr(\hat{G}_k\neq G_k|\mathcal{C}_n)]\cr
&\leq \sum_{j=B}^{A-B+1} \left(Q\left(\frac{\sqrt{\alpha-j+1}}{\sqrt{V}}L\right)+\frac{\tau_1+\tau_2}{\sqrt{(\alpha-j+1)n}} \right)\cr
&\qquad\qquad+\sum_{j=A-B+2}^{k} \left(Q\left(\frac{\sqrt{T_k-j+1}}{\sqrt{V}}L\right)+\frac{\tau_1+\tau_2}{\sqrt{(T_k-j+1)n}}\right) \\
&\leq \sum_{j=\alpha-A+B}^{\alpha-B+1} \left(Q\left(\frac{\sqrt{j}}{\sqrt{V}}L\right)+\frac{\tau_1+\tau_2}{\sqrt{jn}} \right)+\sum_{j=T}^{T_k-A+B-1} \left(Q\left(\frac{\sqrt{j}}{\sqrt{V}}L\right)+\frac{\tau_1+\tau_2}{\sqrt{jn}}\right) \label{eqn:sumall}\\
&\overset{(a)}{\leq} \sum_{j=B}^{A-B+1} \left(Q\left(\frac{\sqrt{j}}{\sqrt{V}}L\right)+\frac{\tau_1+\tau_2}{\sqrt{jn}} \right)+\sum_{j=T}^{A-B+T} \left(Q\left(\frac{\sqrt{j}}{\sqrt{V}}L\right)+\frac{\tau_1+\tau_2}{\sqrt{jn}}\right),\label{eqn:sumall_ub}
\end{align}
where $(a)$ is because if $\alpha=T_k$, which implies $T_k\leq A$, the RHS of \eqref{eqn:sumall} is upper-bounded as follows: 
\begin{align}
\mbox{RHS of \eqref{eqn:sumall}}&=\sum_{j=T}^{T_k-B+1} \left(Q\left(\frac{\sqrt{j}}{\sqrt{V}}L\right)+\frac{\tau_1+\tau_2}{\sqrt{jn}}\right)\\
&\leq\sum_{j=T}^{A-B+1} \left(Q\left(\frac{\sqrt{j}}{\sqrt{V}}L\right)+\frac{\tau_1+\tau_2}{\sqrt{jn}}\right), 
\end{align}
and if $\alpha=A$, which implies $A\leq T_k$, the RHS of \eqref{eqn:sumall} is upper-bounded as follows:
\begin{align}
\mbox{RHS of \eqref{eqn:sumall}}&=\sum_{j=B}^{A-B+1} \left(Q\left(\frac{\sqrt{j}}{\sqrt{V}}L\right)+\frac{\tau_1+\tau_2}{\sqrt{jn}} \right)+\sum_{j=T}^{T_k-A+B-1} \left(Q\left(\frac{\sqrt{j}}{\sqrt{V}}L\right)+\frac{\tau_1+\tau_2}{\sqrt{jn}}\right)\\
&\leq \sum_{j=B}^{A-B+1} \left(Q\left(\frac{\sqrt{j}}{\sqrt{V}}L\right)+\frac{\tau_1+\tau_2}{\sqrt{jn}} \right)+\sum_{j=T}^{A-B+T} \left(Q\left(\frac{\sqrt{j}}{\sqrt{V}}L\right)+\frac{\tau_1+\tau_2}{\sqrt{jn}}\right). 
\end{align}
Now, the RHS of \eqref{eqn:sumall_ub} is bounded as follows:
\begin{align}
\mbox{RHS of \eqref{eqn:sumall_ub}}&=\sum_{j=B}^{A-B+1} \left(Q\left(\frac{\sqrt{j}}{\sqrt{V}}L\right)+\frac{\tau_1+\tau_2}{\sqrt{jn}} \right)+\sum_{j=T}^{A-B+T} \left(Q\left(\frac{\sqrt{j}}{\sqrt{V}}L\right)+\frac{\tau_1+\tau_2}{\sqrt{jn}}\right)\cr
&\overset{(a)}{\leq} \sum_{j=B}^{A-B+1} Q\left(\frac{\sqrt{j}}{\sqrt{V}}L\right)  +\sum_{j=T}^{A-B+T} Q\left(\frac{\sqrt{j}}{\sqrt{V}}L\right)+ 4(\tau_1+\tau_2) \sqrt{\frac{A-B+T}{n}}\\
&\overset{(b)}{\leq} \frac{\sqrt{V}}{\sqrt{2\pi B}L}\cdot \frac{\exp\left\{-\frac{L^2B}{2V}\right\}}{1-\exp\left\{-\frac{L^2}{2V}\right\}} +\sum_{j=T}^{A-B+T} Q\left(\frac{\sqrt{j}}{\sqrt{V}}L\right)+ 4(\tau_1+\tau_2) \sqrt{\frac{A-B+T}{n}} \label{eqn:beforeAB}
\end{align}
where $(a)$ is from Lemma \ref{lemma:sum_integral} (with the identification of $f(j) \equiv  1/\sqrt{ j } $), which is relegated to the end of this subsection, and $(b)$ is obtained by applying similar steps as in the proof of Corollary \ref{coro:CL_asymp}.\footnote{Step $(b)$ can be obtained by replacing $T$ by $B$ in the RHS of \eqref{eqn:cf}.}

Now let us choose $A=n^{1-\delta}$ and $B=\frac{V}{L^2}\delta \log n$ for $0<\delta<\frac{1}{2}$. By substituting this choice of $A$ and $B$ into the RHS of \eqref{eqn:CL_msg} and  the RHS of \eqref{eqn:beforeAB}, we obtain
\begin{align}
\log M_n=nC-L\sqrt{n}+O(n^{\delta}\log n)
\end{align}
and
\begin{align}
\E_{\mathcal{C}_n}[\Pr(\hat{G}_k\neq G_k|\mathcal{C}_n)]\leq \sum_{j=T}^{\infty} Q\left(\frac{\sqrt{j}}{\sqrt{V}}L\right)+O(n^{-\delta/2}), \label{eqn:ensemble_CL}
\end{align}
respectively. Due to the symmetry of the decoding procedure, the bound \eqref{eqn:ensemble_CL} holds for $G_k\in \mathcal{P}(q)$ for $q\geq 3$. For $G_k\in \mathcal{P}(1)$, by defining the error events in the same way as for the moderate deviations regime and then applying similar bounding techniques used in the above, it can be verified that 
\begin{align}
\E_{\mathcal{C}_n}[\Pr(\hat{G}_k\neq G_k|\mathcal{C}_n)]&\leq \sum_{j=T}^{T_k} Q\left(\frac{\sqrt{j}}{\sqrt{V}}L\right)+\frac{\tau_1+\tau_2}{\sqrt{jn}}\\
&\leq  \sum_{j=T}^{A-B+T} Q\left(\frac{\sqrt{j}}{\sqrt{V}}L\right)+\frac{\tau_1+\tau_2}{\sqrt{jn}} \\
&\leq \sum_{j=T}^{\infty} Q\left(\frac{\sqrt{j}}{\sqrt{V}}L\right)+O(n^{-\delta/2}).
\end{align}
Hence, there must exist a sequence of codes $\mathcal{C}_n$ that satisfies \eqref{eqn:CL_rate} and \eqref{eqn:CL_err}, which completes the proof. \endproof

The following basic lemma is used in the proof of Theorem \ref{thm:CL}, whose proof is omitted. 
\begin{lemma} \label{lemma:sum_integral}
Assume two integers $a$ and $b$ such that $a\leq b$. If $f(x)$ is monotonically decreasing and integrable on $[a,b]$, we have 
\begin{align}
\sum_{j=a}^b f(j) &\leq \int^{b+1}_a f(x-1)dx\\
&=F(b)-F(a-1),
\end{align}
where $F(x)$ denotes the antiderivative of $f(x)$.  
\end{lemma}

\section{Extensions in the Moderate Deviations Regime} \label{sec:extension}
In this section, we explore interesting variations of the basic streaming setup in Section \ref{sec:model}. For the brevity of the results, we focus on the moderate deviations regime. 

\subsection{Decoding with an erasure option}
Consider the scenario where there is an erasure option at the decoder, i.e., the decoder can output an erasure symbol instead of a message estimate. 
In the presence of an erasure option, there are two types of error events: (i) the decoder declares an erasure and (ii) the decoder outputs an incorrect message, not an erasure. In many applications, the undetected error (the latter event) is more undesirable than an erasure (the former event).
In the following, we define a streaming code with an erasure option by taking into account the undetected error and the total error probabilities separately. 
\begin{definition}[Streaming code with an erasure option]	
An $(n,M,\epsilon, \epsilon',T)$-streaming code with an erasure option consists of 
\begin{itemize}
\item a sequence of messages $\{G_k\}_{k\geq 1}$ each distributed uniformly over $\mathcal{G}:= [1:M]$,
\item a sequence of encoding functions $\phi_k: \mathcal{G}^k\rightarrow \mathcal{X}^n$ that maps the message sequence $G^k\in  \mathcal{G}^k$ to the channel input codeword $\mathbf{X}_k \in \mathcal{X}^n$, and
\item a sequence of decoding functions $\psi_k: \mathcal{Y}^{(k+T-1)n}\rightarrow \mathcal{G}\cup \{0\}$ that maps the channel output sequences $\mathbf{Y}^{k+T-1} \in \mathcal{Y}^{(k+T-1)n}$ to a message estimate $\hat{G}_k \in \mathcal{G}$ or an erasure symbol $\hat{G}_k = 0$,
\end{itemize}
that satisfies  
\begin{align}
\limsup_{N\rightarrow \infty}\sum_{k=1}^N\frac{\Pr(\hat{G}_k\neq G_k)}{N} \leq \epsilon,
\end{align}
i.e., the total error probability does not exceed $\epsilon$, and 
\begin{align}
\limsup_{N\rightarrow \infty}\sum_{k=1}^N\frac{\Pr(\hat{G}_k\neq G_k, \hat{G}_k\neq 0)}{N} \leq \epsilon',
\end{align}
i.e., the  undetected error probability does not exceed $\epsilon'$. 
\end{definition}

The following theorem presents upper bounds on the undetected error and the total error probabilities. The proof of this theorem is provided in Appendix \ref{appendix:erasure}.  
\begin{theorem} \label{thm:erasure}
Consider a DMC $(\mathcal{X},\mathcal{Y}, \{W(y|x): x\in \mathcal{X}, y\in \mathcal{Y}\} )$ with $V>0$ and any sequence of integers $M_n$ such that $\log M_n=nC-n\rho_n$, where $\rho_n>0, \rho_n\rightarrow 0$ and $n\rho_n^2\rightarrow \infty$. For any $0<\gamma<1$, there exists a sequence of $(n, M_n, \epsilon_n, \epsilon'_n, T)$-streaming codes with an erasure option such that 
\begin{align}
\limsup_{n\rightarrow \infty} \frac{1}{n\rho_n^2} \log \epsilon_n &\leq -\frac{T(1-\gamma)^2}{2V}. \label{eqn:eras_tot}\\
\limsup_{n\rightarrow \infty} \frac{1}{n\rho_n} \log \epsilon'_n &\leq -T\gamma. \label{eqn:eras_und}
\end{align}
\end{theorem}
Theorem \ref{thm:erasure} indicates that for our proposed scheme, the undetected error probability decays much faster than the total error probability, i.e., the exponent of the undetected error probability is the order of $n\rho_n$, whereas that of the total error probability is the order of $n\rho_n^2$. We note that when $T=1$ and $\rho_n=an^{-t}$ for $a>0$ and $0<t<1/2$, Theorem \ref{thm:erasure} reduces to  \cite[Theorem 1]{HayashiTan:15}. In the streaming setup, both the exponents of the total error and the undetected error probabilities improve over the block coding or non-streaming setup in \cite[Theorem 1]{HayashiTan:15} by factors of $T$.

\subsection{Decoding with average delay constraint}
We note that the decoding delay is assumed to be fixed to $T$ up to this point. In this subsection, we relax this constraint by requiring the \emph{average}  decoding delay not to exceed $T$. 
A streaming code with average delay constraint is defined as follows: 
\begin{definition}[Streaming code with average delay constraint]	
An $(n,M,\epsilon, T)$-streaming code with average delay constraint consists of 
\begin{itemize}
\item a sequence of messages $\{G_k\}_{k\geq 1}$ each distributed uniformly over $\mathcal{G}:= [1:M]$,
\item a sequence of encoding functions $\phi_k: \mathcal{G}^k\rightarrow \mathcal{X}^n$ that maps the message sequence $G^k\in  \mathcal{G}^k$ to the channel input codeword $\mathbf{X}_k \in \mathcal{X}^n$, and
\item a sequence of decoding functions $\psi_k: \mathcal{Y}^{kn}\rightarrow (\mathcal{G}\cup \{0\})^k$ that maps the channel output sequences $\mathbf{Y}^{k} \in \mathcal{Y}^{kn}$  to a message estimate $\hat{G}_{k,j} \in \mathcal{G}$ or an erasure symbol $\hat{G}_{k,j} = 0$ for every $j\in [1:k]$
\end{itemize}
 that satisfies 
\begin{align}
\limsup_{N\rightarrow \infty}\sum_{k=1}^N\frac{\Pr(\hat{G}_{k+D_k-1,k}\neq G_{k})}{N} \leq \epsilon 
\end{align} and 
\begin{align}
\limsup_{N\rightarrow \infty}\sum_{k=1}^N\frac{\E[D_k]}{N}\leq T,
\end{align}
where $D_k:=\min\{d: \hat{G}_{k+d-1,k}\neq 0 \}$ for $k\in \bbN$ denotes the (random) decoding delay of the $k$-th message.\footnote{Note that message $G_k$ is required to be decoded at the end of every block on and after the $k$-th block in this definition. One may wonder why the decoder does not stop decoding $G_k$ after it outputs an estimate of $G_k$, not an erasure. We note that our definition includes such a operation as a special case by letting the decoder simply fix the estimate of $G_k$ once it outputs a message estimate. 
}  
\end{definition}
For block channel coding with feedback, it is known that the error exponent can be significantly improved by allowing variable decoding delay, e.g., \cite{Forney:68}.  For streaming setup, the following theorem, which is proved in Appendix~\ref{appendix:variable}, shows that such an improvement can be obtained in the absence of feedback.  
\begin{theorem} \label{thm:variable}
Consider a DMC $(\mathcal{X},\mathcal{Y}, \{W(y|x): x\in \mathcal{X}, y\in \mathcal{Y}\} )$ with $V>0$ and any sequence of integers $M_n$ such that $\log M_n=nC-n\rho_n$, where $\rho_n>0, \rho_n\rightarrow 0$ and $n\rho_n^2\rightarrow \infty$. For any $T\in \bbN$, there exists a sequence of $(n, M_n, \epsilon_n, T_n)$-streaming codes with average delay constraint such that 
\begin{align}
\lim_{n\rightarrow \infty} T_n&=T\\
\limsup_{n\rightarrow \infty} \frac{1}{n\rho_n} \log \epsilon_n &\leq -T. \label{eqn:var}
\end{align}
\end{theorem}
We note that the exponent of the error probability $\epsilon_n$ is of the order $n\rho_n$ (instead of $n\rho_n^2$ as in \eqref{eqn:eras_tot}), and hence it is improved tremendously by allowing variable decoding delay. 

\subsection{Alternating message rates}
We note that the rates of the messages are assumed to be fixed across time thus far. In many practical streaming applications, however, a stream of data packets does not have a constant rate. For example, in the MPEC standard for video coding,  I frames have higher rates than P frames in general.\footnote{I frame is an intra-coded picture that is coded using only information present in the picture itself and P frame is a predictive-coded picture that is coded with inter-frame prediction from previous frames.} Similarly, in audio coding, voice packets have higher rates than silent packets. In this subsection, to obtain useful insights when the message rates vary across time, we assume a simple example where the rate of the messages in odd block indices and the rate of the messages in even block indices are different. A streaming code with alternating message rates is defined as follows: 
\begin{definition}[Streaming code with alternating message rates]	
An $(n,M,r,\epsilon_1, \epsilon_2,T)$-streaming code with alternating message rates\footnote{It is assumed that $0<r<1$.} consists of 
\begin{itemize}
\item a sequence of messages $\{G_k\}_{k\geq 1}$ where message $G_{2j-1}$ for $j\in \bbN$ is distributed uniformly over $\mathcal{G}_1:= [1:M^r]$ and message $G_{2j}$ for $j\in \bbN$ is distributed uniformly over $\mathcal{G}_2:= [1:M^{2-r}]$,
\item a sequence of encoding functions $\phi_k: \mathcal{G}_1^{\lceil k/2 \rceil}\times \mathcal{G}_2^{\lfloor k/2 \rfloor } \rightarrow \mathcal{X}^n$ that maps the message sequence $\{G_{2j-1}:j\in [1: \lceil k/2 \rceil] \} \cup \{G_{2j}: j\in [1:\lfloor k/2 \rfloor] \} \in  \mathcal{G}_1^{\lceil k/2 \rceil}\times \mathcal{G}_2^{\lfloor k/2 \rfloor }$ to the channel input codeword $\mathbf{X}_k \in \mathcal{X}^n$, and
\item a sequence of decoding functions $\psi_k: \mathcal{Y}^{(k+T-1)n}\rightarrow \mathcal{G}_{2-(k\mathrm{~mod}2) }$ that maps the channel output sequences $\mathbf{Y}^{k+T-1} \in \mathcal{Y}^{(k+T-1)n}$ to the message estimate $\hat{G}_k\in \mathcal{G}_{2-(k\mathrm{~mod}2)}$,
\end{itemize}
that satisfies  
\begin{align}
\limsup_{N\rightarrow \infty}\sum_{j=1}^N\frac{\Pr(\hat{G}_{2j-1}\neq G_{2j-1})}{N} \leq \epsilon_1,
\end{align}
and 
\begin{align}
\limsup_{N\rightarrow \infty}\sum_{j=1}^N\frac{\Pr(\hat{G}_{2j}\neq G_{2j})}{N} \leq \epsilon_2.
\end{align}
\end{definition}
We note that the error probabilities are considered separately for the messages in the odd block indices and for the messages in the even block indices. The following theorem, which is proved in Appendix \ref{appendix:alternate}, gives achievability bounds when the message rates are alternating.

\begin{theorem} \label{thm:alternate}
Consider a DMC $(\mathcal{X},\mathcal{Y}, \{W(y|x): x\in \mathcal{X}, y\in \mathcal{Y}\} )$ with $V>0$ and any sequence of integers $M_n$ such that $\log M_n=nC-n\rho_n$, where $\rho_n>0, \rho_n\rightarrow 0$ and $n\rho_n^2\rightarrow \infty$. For any $r\in (0,1)$, there exists a sequence of $(n,M_n,r,\epsilon_{1,n}, \epsilon_{2,n},T)$-streaming codes with alternating message rates such that 
\begin{align}
\limsup_{n\rightarrow \infty} \frac{1}{n\rho_n^2} \log \epsilon_{1,n} &\leq -\frac{T+1}{2V} \\
\limsup_{n\rightarrow \infty} \frac{1}{n\rho_n^2} \log \epsilon_{2,n} &\leq -\frac{T-1}{2V} 
\end{align}
if $T$ is odd, and 
\begin{align}
\limsup_{n\rightarrow \infty} \frac{1}{n\rho_n^2} \log \epsilon_{1,n} &\leq -\frac{T}{2V} \\
\limsup_{n\rightarrow \infty} \frac{1}{n\rho_n^2} \log \epsilon_{2,n} &\leq -\frac{T}{2V} 
\end{align}
if $T$ is even. 
\end{theorem}
Theorem \ref{thm:alternate} indicates that for our coding strategy, the average of the moderate deviations constants does not change even if the rates are alternating. We note that the moderate deviations constants of the odd and even messages are asymmetric if $T$ is odd and they are symmetric otherwise. To illustrate the intuition behind this, let us consider the most recent $T$ blocks, which dominates the error probability. 
If $T$ is odd, the average message rate in those $T$ blocks depends on whether the current target message index is odd or even and this leads to the asymmetry. On the other hand, if $T$ is even, it is fixed to $\log M_n$ regardless of the current target message index. More details are provided in Appendix \ref{appendix:alternate}.

\section{Conclusion}  \label{sec:conclusion}
In this paper, we studied the fundamental interplay between the rate and error probability for a streaming setup with a decoding delay of $T$ blocks. In the moderate deviations regime, the moderate deviations constant was shown to improve by at least a factor of $T$. We proposed a coding technique with infinite memory such that all the previous and fresh messages are jointly encoded in each block. On the other hand, in the central limit regime, the second-order coding rate was shown to improve by approximately a factor of $\sqrt{T}$ for a wide range of channel parameters. To ensure that the summation of Berry-Esseen constants (e.g., the last terms in the RHS of \eqref{eqn:berry_1}-\eqref{eqn:berry_3}) does not diverge in the error analysis, we proposed a coding technique with truncated memory such that the encoding and decoding memories do not grow with the block index. 
Furthermore, we generalized the moderate deviations result in various directions. We first considered a scenario with an erasure option at the decoder and showed that both the exponents of the total error and the undetected error probabilities improve by factors of $T$. By utilizing the erasure option, we showed 
that the exponent of the total error probability can be improved to that of the undetected error probability (in the order sense) at the expense of a variable decoding delay. Finally, we considered a scenario where the message rates are alternating and showed that the same average  moderate deviations constant as the case of constant rate can be obtained. 
We note that all of our encoding strategies do not depend on $T$. Hence, our coding techniques are directly applicable for multicast scenario where a sender transmits a common stream of data packets to multiple receivers with possibly different decoding constraints. 

Let us conclude with a final remark on proving a converse in our streaming setup. Our problem appears to  be closely related to the bit-wise unequal protection (UEP) problem in the sense that we need to capture the tension that arises when a common channel is used for more than one messages with individual error criteria.\footnote{We note that there are two types of UEP problems, i.e., bit-wise and message-wise UEP, but our streaming setup is more related to the bit-wise UEP. For example, the message-wise UEP problem studied by Shkel et al.~\cite{ShkelTanDraper:15} simply considers partitioning a  \emph{single} message set into several sub-message sets with different error  probability    requirements.} For the seemingly simpler bit-wise UEP problem \cite{BoradeNakibogluZheng:09} for the block channel coding with the \emph{same} decoding deadline, however, tight characterizations of various  asymptotic fundamental limits (e.g., error exponents) remain  challenging open problems in general.  This indicates that a highly-nontrivial converse technique, perhaps along the lines of Sahai's work \cite{Sahai:08}, would be needed for our streaming setup where the messages have \emph{different} decoding deadlines. 

\appendices

\section{Proof of Corollary \ref{coro:CL_asymp}} \label{appendix:asymptotic}
Let $\mu := L/\sqrt{V}$. Note that Corollary \ref{coro:CL_asymp} is proved if we show 
\begin{align}
\sum_{j=T}^{\infty}Q\left( \mu \sqrt{j }\right)\leq c_{L,V,T} Q(\mu\sqrt{T}).
\end{align}
To that end, we use the following bounds on the $Q$-function: 
\begin{align}
\frac{x\phi(x)}{1+x^2}\le Q(x) \le\frac{\phi(x)}{x} \qquad\forall\, x > 0,
\end{align}
where $\phi(x):=\frac{1}{\sqrt{2\pi}}\exp\{-\frac{x^2}{2}\}$. Then, we have
\begin{align}
\sum_{j=T}^{\infty}Q\left( \mu \sqrt{j }\right)&\leq \sum_{j=T}^{\infty}\frac{\phi\left( \mu \sqrt{ j }\right) }{\mu\sqrt{j }}\\
&\le \frac{1}{\mu\sqrt{ T}}\sum_{j=T}^{\infty} \frac{1}{\sqrt{2\pi}}\exp\{-\mu^2 j/2\} \\
&= \frac{1}{\mu\sqrt{T}} \cdot  \frac{1}{\sqrt{2\pi}} \frac{\exp\{-\mu^2 T/2\}}{1-\exp\{-\mu^2/2\}}\label{eqn:cf}\\
&\le \frac{1+\mu^2 T}{\mu^2 T} \cdot  \frac{ Q\left(\mu \sqrt{T }\right)  }{1-\exp\{-\mu^2/2\}}\\
&=c_{L,V,T} Q(\mu\sqrt{T}),
\end{align}
which completes the proof. \endproof

\section{Proof of Lemma \ref{lemma:nonasymptotic_MD}} \label{appendix:finite_MD}
Fix $n\in\bbN$ and $s\geq  0$. Then, we have
\begin{align}
\Pr\left(\frac{1}{n}\sum_{l=1}^n Z_l \geq \eps_n\right) &\leq \Pr\left( \exp \Big\{s\sum_{l=1}^n Z_l\Big\} \geq \exp\{ns\eps_n\} \right)\\
 &\overset{(a)}{\le} \exp\{-ns\eps_n\}\E \left[ \exp\Big\{s\sum_{l=1}^n Z_l\Big\} \right]  \\
 &\overset{(b)}{=}\exp\left\{-n \left(s\eps_n-\log\E [\exp\{sZ_1\}]\right)\right\} \\ 
  &=\exp\left\{-n \left( s\eps_n-h(s)\right)\right\} \label{eqn:chernoff} .
\end{align}
where $(a)$ follows from Markov's inequality and $(b)$ follows from the independence of $Z_l$'s.

The third-order Taylor series expansion of the cumulant generating function $h(s)$ can be written as
\begin{equation}
h(s) = h(0) + s h'(0) +\frac{s^2}{2} h''(0) + \frac{s^3}{6}  h'''(\tils) \label{eqn:cgf_expand}
\end{equation}
for some $\tils\in [0,s]$. It is easy to check that $h(0)=0$, $h'(0)=\E[Z_1] = 0$ and $h''(0) = \var[Z_1]=\sigma^2$.   Now, we take 
\begin{equation}
s:=\frac{\eps_n}{\sigma^2}.
\end{equation}
Plugging this into  \eqref{eqn:chernoff} and \eqref{eqn:cgf_expand}  yields
\begin{align}
\Pr\left(\frac{1}{n}\sum_{l=1}^n Z_l \geq \eps_n\right) &  \le \exp\left\{ -n \left(  \frac{\eps_n^2}{\sigma^2}-\frac{\eps_n^2}{2\sigma^2} -\frac{\eps_n^3}{6\sigma^6} h'''( \tils)\right)\right\} \\
&  \le \exp\left\{ -n \left(  \frac{\eps_n^2}{2\sigma^2}  -\frac{\eps_n^3}{6\sigma^6}  K \right)\right\} ,
\end{align}
where the final inequality holds for all $n$ sufficiently large since $\eps_n\to 0$ and $\tils\to 0$ as $n\to\infty$ and thus $| h'''(\tils) |\le K$.  \endproof
%

\section{A chain of equalities}\label{appendix:chain}
The following chain of equalities is used in the proof Theorem \ref{thm:CL}.
\begin{align}
&\Pr\left(\sum_{l=1}^{n(\alpha-B+1)} i(X_l;\bar{Y}_l)> \alpha\cdot \log M_n\right)\cr
&=\sum_{s=1}^{n(\alpha-B+1)} \sum_{x_s, \bar{y}_s}\left(\prod_{t=1}^{n(\alpha-B+1)}P_X(x_t)P_XW(\bar{y}_t)\right)\cdot \mathbbm{1}\Big\{\sum_{l=1}^{n(\alpha-B+1)} i(x_l;\bar{y}_l)> \alpha\cdot \log M_n\Big\}\label{eqn:chain1}\\
&=\sum_{s=1}^{n(\alpha-B+1)} \sum_{x_s, \bar{y}_s}\left(\prod_{t=1}^{n(\alpha-B+1)}P_X(x_t)W(\bar{y}_t|x_t)\frac{P_XW(\bar{y}_t)}{W(\bar{y}_t|x_t)}\right)\cr
&\qquad\qquad\qquad\qquad\qquad\qquad\qquad\qquad\qquad\cdot \mathbbm{1}\Big\{\sum_{l=1}^{n(\alpha-B+1)} i(x_l;\bar{y}_l)> \alpha\cdot \log M_n\Big\}\\
&=\sum_{s=1}^{n(\alpha-B+1)} \sum_{x_s, \bar{y}_s}\left(\prod_{t=1}^{n(\alpha-B+1)}P_X(x_t)W(\bar{y}_t|x_t)\right)\cdot \exp\left\{-\sum_{l=1}^{n(\alpha-B+1)} i(x_{l};\bar{y}_{l})\right\}\cr
&\qquad\qquad\qquad\qquad\qquad\qquad\qquad\qquad\qquad\cdot \mathbbm{1}\Big\{\sum_{l=1}^{n(\alpha-B+1)} i(x_l;\bar{y}_l)> \alpha\cdot \log M_n\Big\}\\
&= \E\left[\exp\left\{-\sum_{l=1}^{n(\alpha-B+1)} i(X_l;Y_l)\right\}\right. \cdot \left.\mathbbm{1}\Big\{\sum_{l=1}^{n(\alpha-B+1)} i(X_l;Y_l)>\alpha \log M_n\Big\}\right]. \label{eqn:chain4}
\end{align}

\section{Proof of Theorem \ref{thm:erasure} }\label{appendix:erasure}
Consider a DMC $(\mathcal{X},\mathcal{Y}, \{W(y|x): x\in \mathcal{X}, y\in \mathcal{Y}\} )$ with $V>0$ and any sequence of integers $M_n$ such that $\log M_n=nC-n\rho_n$, where $\rho_n>0, \rho_n\rightarrow 0$ and $n\rho_n^2\rightarrow \infty$. We denote by $P_X$ an input distribution that achieves the dispersion \eqref{eqn:dispersion}. Fix $0<\gamma<1$. 

The encoding procedure is the same as that for the basic streaming setup in Section \ref{subsec:MD_pf}. Let us consider the decoding of $G_k$ at the end of block $T_k$. The decoding procedure is modified from that  for the basic streaming setup in Section \ref{subsec:MD_pf} as follows:  
\begin{itemize}
\item The decoding test \eqref{eqn:dec_rule} is modified as follows: 
\begin{align}
i(\mathbf{x}_{[j:T_k]}(\hat{G}_{T_k, [1:j-1]}, g_{[j:T_k]}), \mathbf{y}_{[j:T_k]})>(T_k-j+1) \cdot (\log M_n+ \gamma n\rho_n),  \label{eqn:dec_rule_erasure}
\end{align}
i.e., the threshold value  is increased proportional to $\gamma$. 

\item If there is none or more than one $g_j$ that satisfies the decoding test \eqref{eqn:dec_rule_erasure} for some $g_{[j+1:T_k]}$, the decoder declares an erasure, i.e., $\hat{G}_{k}=0$, and terminates the decoding procedure. 
\end{itemize}
Similarly as in Section \ref{subsec:MD_pf}, we first consider the probability of error averaged over random codebook $\mathcal{C}_n$. The error event $\{\hat{G}_k\neq G_k\}$ for $k\in \bbN$ happens only if at least one of the following $2k$ events occurs: 
\begin{align}
\mathcal{E}_{k,j}'&:= \{
i(\mathbf{X}_{[j:T_k]}(G^{T_k}), \mathbf{Y}_{[j:T_k]})\leq (T_k-j+1) \cdot (\log M_n+ \gamma n\rho_n)
\},~ j\in [1:k] \label{eqn:err1_ers}\\
\tilde{\mathcal{E}}_{k,j}'&:= \{i(\mathbf{X}_{[j:T_k]}(G^{j-1},g_{[j:T_k]}), \mathbf{Y}_{[j:T_k]})> (T_k-j+1) \cdot (\log M_n+ \gamma n\rho_n) \cr
&~~~~~~~~~~~~~~~~~~~~~~~~~~~~~~~~~~~~ \mbox{ for some }g_{[j:T_k]} \mbox{ such that }g_{j} \neq G_{j} \},~ j\in [1:k]. \label{eqn:err2_ers}
\end{align}
We note that \eqref{eqn:err1_ers} and \eqref{eqn:err2_ers} are obtained by replacing $\log M_n$ by $\log M_n+ \gamma n\rho_n$ in \eqref{eqn:err1} and \eqref{eqn:err2}, respectively. 
Then, we have
\begin{align}
\E_{\mathcal{C}_n}[\Pr(\hat{G}_k\neq G_k|\mathcal{C}_n)]&\leq \sum_{j=1}^k \left(\Pr(\mathcal{E}_{k,j}')+\Pr(\tilde{\mathcal{E}}_{k,j}') \right).
\end{align}
On the other hand, the undetected error event $\{\hat{G}_k\neq G_k, \hat{G}_k\neq 0\}$ has the following relationship: 
\begin{align}
\{\hat{G}_k\neq G_k, \hat{G}_k\neq 0\}&\subseteq \{\hat{G}_{T_k, [1:k]}\neq G_{[1:k]}, \hat{G}_k\neq 0\} \\
&= \cup_{j\in [1:k]}\{\hat{G}_{T_k, [1:j-1]}= G_{[1:j-1]}, \hat{G}_{T_k, j}\neq G_{j}, \hat{G}_k\neq 0\}. 
\end{align}

Hence, the undetected error probability is bounded as follows: 
\begin{align}
\E_{\mathcal{C}_n}[\Pr(\hat{G}_k\neq G_k, \hat{G}_k\neq 0|\mathcal{C}_n)]&\leq \sum_{j=1}^k \Pr(\hat{G}_{T_k, [1:j-1]}= G_{[1:j-1]}, \hat{G}_{T_k, j}\neq G_{j}, \hat{G}_k\neq 0)\\
&\leq \sum_{j=1}^k \Pr(\tilde{\mathcal{E}}_{k,j}').
\end{align}

Now, for $j\in [1:k]$, let us bound $\Pr(\mathcal{E}_{k,j}')$ and $\Pr(\tilde{\mathcal{E}}_{k,j}')$. Similarly as in Section \ref{subsec:MD_pf},  $(X_l,Y_l, \bar{Y}_l)$'s denote i.i.d. random variables each generated according to $P_X(x_l)W(y_l|x_l)$$P_XW(\bar{y}_l)$ in the following. First, we have  
\begin{align*}
\Pr(\mathcal{E}_{k,j}')&\leq \Pr\left(\sum_{l=1}^{n (T_k-j+1) } i(X_l;Y_l)\leq (T_k-j+1) \cdot (\log M_n+ \gamma n\rho_n)\right)\\
&\leq \Pr\left(\sum_{l=1}^{n(T_k-j+1)} i(X_l;Y_l)\leq (T_k-j+1)n(C-(1-\gamma)\rho_n)\right)\cr
&\overset{(a)}{\leq} \exp\left\{-(T_k-j+1)n\left(\frac{(1-\gamma)^2 \rho_n^2}{2V}-(1-\gamma)^3 \rho_n^3 \tau\right)\right\}
\end{align*}
for sufficiently large $n$, where $\tau$ is some non-negative constant dependent only on the input distribution $P_X(x)$ and channel statistics $W(y|x)$ and $(a)$ is from Lemma~\ref{lemma:nonasymptotic_MD} in Section \ref{subsec:MD_pf}.

Next, we have 
\begin{align}
\Pr(\tilde{\mathcal{E}}_{k,j}')&\leq M_n^{T_k-j+1} \cdot \Pr\left(\sum_{l=1}^{n (T_k-j+1) } i(X_l;\bar{Y}_l)>  (T_k-j+1) \cdot( \log M_n+\gamma n\rho_n)\right)\\
&\overset{(a)}{=} M_n^{T_k-j+1}\cdot \E\left[\exp\left\{-\sum_{l=1}^{n(T_k-j+1)} i(X_l;Y_l)\right\} \right.\cr
&\qquad \qquad \qquad\qquad \qquad\cdot \left.\mathbbm{1}\Big\{\sum_{l=1}^{n(T_k-j+1)} i(X_l;Y_l)>(T_k-j+1) \cdot( \log M_n+\gamma n\rho_n)\Big\}\right]\\
&\leq M_n^{T_k-j+1}\cdot \exp\left\{-(T_k-j+1) \cdot( \log M_n+\gamma n\rho_n)\right\}\\
&=\exp\left\{-(T_k-j+1)\gamma n\rho_n\right\},
\end{align}
where $(a)$ is obtained by applying a chain of equalities similar to that in Appendix \ref{appendix:chain}.

Hence, we obtain
\begin{align}
\E_{\mathcal{C}_n}[\Pr(\hat{G}_k\neq G_k|\mathcal{C}_n)]&\leq \sum_{j=1}^k \Big( \exp\left\{-(T_k-j+1)n\rho_n^2(1-\gamma)^2\left(\frac{1}{2V}-(1-\gamma) \rho_n \tau\right)\right\}\cr
&\qquad \qquad \qquad +\exp\left\{-(T_k-j+1)n\gamma\rho_n\right\} \Big)\\
&\leq \sum_{j=T}^{T_k} \left( \exp\left\{-jn\rho_n^2(1-\gamma)^2\left(\frac{1 }{2V}-(1-\gamma) \rho_n \tau\right)\right\}+\exp\left\{-jn\gamma \rho_n\right\}\right) \\
&\leq \frac{\exp\left\{-Tn\rho_n^2(1-\gamma)^2\left(\frac{1 }{2V}-(1-\gamma)\rho_n \tau\right)\right\}}{1-\exp\{-n\rho_n^2(1-\gamma)^2\left(\frac{1 }{2V}-(1-\gamma) \rho_n \tau\right)\}}+\frac{\exp\left\{-Tn\gamma\rho_n\right\}}{1-\exp\left\{-n\gamma \rho_n\right\}} \label{eqn:averg_total}
\end{align}
and
\begin{align}
\E_{\mathcal{C}_n}[\Pr(\hat{G}_k\neq G_k, \hat{G}_k\neq 0|\mathcal{C}_n)]&\leq\frac{\exp\left\{-Tn\gamma\rho_n\right\}}{1-\exp\left\{-n\gamma \rho_n\right\}} \label{eqn:averg_undetec}
\end{align}
for sufficiently large $n$. 

To show the existence of a deterministic code, we apply Markov's inequality as follows\footnote{Such a technique of applying  Markov's inequality to derandomize the code was used in the proof of \cite[Theorem 1]{HayashiTan:15}.}: 
\begin{align}
&\Pr\left(\limsup_{N\rightarrow \infty} \sum_{k=1}^N\frac{\Pr(\hat{G}_k\neq G_k|\mathcal{C}_n)}{N}>2\limsup_{N\rightarrow \infty} \sum_{k=1}^N\frac{\E_{\mathcal{C}_n}[\Pr(\hat{G}_k\neq G_k|\mathcal{C}_n)]   }{N}\right)<\frac{1}{2} \\
&\Pr\left(\limsup_{N\rightarrow \infty} \sum_{k=1}^N\frac{\Pr(\hat{G}_k\neq G_k, \hat{G}_k\neq 0|\mathcal{C}_n)}{N}>2\limsup_{N\rightarrow \infty} \sum_{k=1}^N\frac{\E_{\mathcal{C}_n}[\Pr(\hat{G}_k\neq G_k, \hat{G}_k\neq 0|\mathcal{C}_n)]}{N}   \right)<\frac{1}{2}. 
\end{align}
Then, from the union bound, we have 
\begin{align}
&\Pr\Big(\limsup_{N\rightarrow \infty} \sum_{k=1}^N\frac{\Pr(\hat{G}_k\neq G_k|\mathcal{C}_n)}{N}>2\limsup_{N\rightarrow \infty} \sum_{k=1}^N\frac{\E_{\mathcal{C}_n}[\Pr(\hat{G}_k\neq G_k|\mathcal{C}_n)]   }{N} \mbox{ or }\cr
& \limsup_{N\rightarrow \infty} \sum_{k=1}^N\frac{\Pr(\hat{G}_k\neq G_k, \hat{G}_k\neq 0|\mathcal{C}_n)}{N}>2\limsup_{N\rightarrow \infty} \sum_{k=1}^N\frac{\E_{\mathcal{C}_n}[\Pr(\hat{G}_k\neq G_k, \hat{G}_k\neq 0|\mathcal{C}_n)]}{N}   \Big)<1.
\end{align}
Therefore, there must exist a sequence of codes $\mathcal{C}_n$ that satisfies 
\begin{align}
\limsup_{N\rightarrow \infty} \sum_{k=1}^N\frac{\Pr(\hat{G}_k\neq G_k|\mathcal{C}_n)}{N}&\leq 2 \exp\left\{-n\rho_n^2\Big((1-\gamma)^2\frac{T}{2V}+o(1)\Big)\right\} 
\end{align}
and
\begin{align}
\limsup_{N\rightarrow \infty} \sum_{k=1}^N\frac{\Pr(\hat{G}_k\neq G_k, \hat{G}_k\neq 0|\mathcal{C}_n)}{N}&\leq 2\exp\left\{-n\rho_n(T\gamma +o(1))\right\},
\end{align}
which completes the proof. \endproof

\section{Proof of Theorem \ref{thm:variable} }\label{appendix:variable}
Consider a DMC $(\mathcal{X},\mathcal{Y}, \{W(y|x): x\in \mathcal{X}, y\in \mathcal{Y}\} )$ with $V>0$ and any sequence of integers $M_n$ such that $\log M_n=nC-n\rho_n$, where $\rho_n>0, \rho_n\rightarrow 0$ and $n\rho_n^2\rightarrow \infty$. We denote by $P_X$ an input distribution that achieves the dispersion \eqref{eqn:dispersion}. Fix $T\in \bbN$ and $0<\gamma<1$. 

The encoding procedure is the same as that for the basic streaming setup in Section \ref{subsec:MD_pf}. Let us consider the decoding of message $G_k$  at the end of block $k+d-1$ for $d\in \bbN$.\footnote{We note that in the definition of a streaming code with average delay constraint, the decoder decodes $G_k$ at the end of \emph{every} block $k+d-1$ for $d\in\bbN$.}  If $d\in [1:T-1]$, the decoder outputs $\hat{G}_{k+d-1,k}=0$. For $d\geq T$, the decoder outputs a message estimate $\hat{G}_{k+d-1,k} \in \mathcal{G}$ or an erasure symbol $\hat{G}_{k+d-1,k} = 0$ according to the same decoding rule illustrated in Appendix \ref{appendix:erasure} with delay $d$. 
Then, the error probability of $G_k$ after the random decoding delay $D_k=\min\{d: \hat{G}_{k+d-1,k}\neq 0 \}$ (averaged over the random codebook generation) is bounded as follows: 
\begin{align}
\E_{\mathcal{C}_n}\big[\Pr(\hat{G}_{k+D_k-1,k}\neq G_{k}|\mathcal{C}_n)\big]&=\sum_{d=T}^{\infty} \E_{\mathcal{C}_n}\big[\Pr(D_k=d, \hat{G}_{k+d-1,k}\neq G_{k}, \hat{G}_{k+d-1,k}\neq 0|\mathcal{C}_n)\big]\\
&\leq \sum_{d=T}^{\infty} \E_{\mathcal{C}_n}\big[\Pr(\hat{G}_{k+d-1,k}\neq G_{k}, \hat{G}_{k+d-1,k}\neq 0|\mathcal{C}_n)\big]\\
&\overset{(a)}{\leq} \sum_{d=T}^{\infty}\frac{\exp\left\{-dn\gamma \rho_n\right\}}{1-\exp\left\{-n\gamma \rho_n\right\}}\\
&\leq \frac{\exp\left\{-Tn\gamma\rho_n\right\}}{(1-\exp\left\{-n\gamma\rho_n\right\})^2},
\end{align}
where $(a)$ is from the upper bound \eqref{eqn:averg_undetec} on the undetected error probability with delay $d$. 

On the other hand, the excess of delay averaged over the random codebook generation is bounded as 
\begin{align}
\E_{\mathcal{C}_n}[D_k-T|\mathcal{C}_n] &=  \Pr(D_k=T+1|\mathcal{C}_n)+2 \Pr(D_k=T+2|\mathcal{C}_n)+\cdots\\
&=\Pr(\hat{G}_{k+T-1,k}=0, \hat{G}_{k+T,k}\neq 0|\mathcal{C}_n)\cr
&\qquad \qquad \qquad+2 \Pr(\hat{G}_{k+T-1,k}=0, \hat{G}_{k+T,k}=0,  \hat{G}_{k+T+1,k}\neq 0|\mathcal{C}_n)+\cdots\\
&\leq \sum_{d=T+1}^{\infty} (d-T)\cdot \Pr\left(\hat{G}_{k+d-2,k} = 0|\mathcal{C}_n\right) \\
&\leq \sum_{d=T+1}^{\infty} (d-T)\cdot \Pr \left(\hat{G}_{k+d-2,k} \neq G_k|\mathcal{C}_n\right)\\
&\overset{(a)}{\leq} \sum_{d=T+1}^{\infty} (d-T)\cdot \Big( \frac{\exp\left\{-(d-1)n\rho_n^2(1-\gamma)^2\left(\frac{1 }{2V}-(1-\gamma) \rho_n \tau\right)\right\}}{1-\exp\{-n\rho_n^2(1-\gamma)^2\left(\frac{1 }{2V}-(1-\gamma) \rho_n \tau\right)\}}\cr
&\qquad\qquad\qquad\qquad\qquad +\frac{\exp\left\{-(d-1)n\gamma\rho_n\right\}}{1-\exp\left\{-n\gamma \rho_n\right\}}\Big),\label{eqn:delay_ub}
\end{align}
where $(a)$ is from the upper bound \eqref{eqn:averg_total}  on the total error probability with delay $d-1$. 

By following similar statements using Markov's inequality in the proof of Theorem \ref{thm:erasure}, we can obtain
\begin{align}
&\Pr\Big(\limsup_{N\rightarrow \infty} \sum_{k=1}^N\frac{\Pr(\hat{G}_{k+D_k-1,k}\neq G_{k}|\mathcal{C}_n)}{N}>2\limsup_{N\rightarrow \infty} \sum_{k=1}^N\frac{\E_{\mathcal{C}_n}[\Pr(\hat{G}_{k+D_k-1,k}\neq G_{k}|\mathcal{C}_n)]   }{N}\cr
&\qquad\qquad\qquad\qquad \mbox{ or } \limsup_{N\rightarrow \infty} \sum_{k=1}^N\frac{\E[D_k|\mathcal{C}_n]}{N}-T>2\limsup_{N\rightarrow \infty} \sum_{k=1}^N\frac{\E_{\mathcal{C}_n}[D_k-T|\mathcal{C}_n]}{N}   \Big)<1.
\end{align}
Therefore, there must exist a sequence of codes $\mathcal{C}_n$ that satisfies
\begin{align}
\limsup_{N\rightarrow \infty} \sum_{k=1}^N\frac{\Pr(\hat{G}_{k+D_k-1,k}\neq G_{k}|\mathcal{C}_n)}{N}&\leq 2  \exp\left\{-n\rho_n(T\gamma +o(1))\right\}\label{eqn:av_err}
\end{align}
and\footnote{By calculating the infinite series in the RHS of \eqref{eqn:delay_ub}, it can be verified that the RHS of  \eqref{eqn:delay_ub} converges to 0 as $n$ tends to infinity. }
\begin{align}
\limsup_{N\rightarrow \infty} \sum_{k=1}^N\frac{\E[D_k|\mathcal{C}_n]}{N}\leq  T+o(1).
\end{align}
We note that \eqref{eqn:av_err} implies 
\begin{align}
\limsup_{n\rightarrow \infty} \frac{1}{n\rho_n }\log \left(\limsup_{N\rightarrow \infty} \sum_{k=1}^N\frac{\Pr(\hat{G}_{k+D_k-1,k}\neq G_{k}|\mathcal{C}_n)}{N}\right)&\leq -T\gamma. 
\end{align}
By taking $\gamma\rightarrow 1$, this completes the proof. \endproof


\section{Proof of Theorem \ref{thm:alternate} }\label{appendix:alternate}
Consider a DMC $(\mathcal{X},\mathcal{Y}, \{W(y|x): x\in \mathcal{X}, y\in \mathcal{Y}\} )$ with $V>0$ and any sequence of integers $M_n$ such that $\log M_n=nC-n\rho_n$, where $\rho_n>0, \rho_n\rightarrow 0$ and $n\rho_n^2\rightarrow \infty$. We denote by $P_X$ an input distribution that achieves the dispersion \eqref{eqn:dispersion}. Fix $0<r<1$.

\subsection{Encoding} \label{subsubsec:MD_encoding}
For each $k\in \bbN$ and $g^k \in \mathcal{G}_1\times \mathcal{G}_2 \times \mathcal{G}_1\times \cdots$, generate $\mathbf{x}_k(g^k)$ in an i.i.d. manner according to $P_X$. The generated codewords constitute the codebook $\mathcal{C}_n$. In block $k$, after observing the true message sequence $G^k$, the encoder sends $\mathbf{x}_k(G^k)$.

\subsection{Decoding} 
Consider the decoding of $G_k$ at the end of block $T_k:=k+T-1$. Similarly as in the decoding procedure in Section \ref{subsec:MD_pf} for the basic streaming setup, the decoder not only decodes $G_k$, but also re-decodes $G_1,\cdots, G_{k-1}$ at the end of block $T_k$. Let $\hat{G}_{T_k,j}$ denote the estimate of $G_j$ at the end of block $T_k$. For alternating message rates, we modify the decoding procedure for the basic streaming setup in Section \ref{subsec:MD_pf}  according to the following rules\footnote{We remind that the messages in odd block indices have a lower rate of $r\log M_n$ and the messages in even block indices have a higher rate of $(2-r)\log M_n$ because $0<r<1$.}: 
\begin{itemize}
\item In the basic streaming setup, we consider the window of blocks $[j:T_k]$  for the decoding of $G_j$ for $j\in [1:k]$ as shown in \eqref{eqn:dec_rule}. Note that the last block index in each decoding window is $T_k$. For alternating message rates, we choose the last block index to be an odd number, to avoid the average message rate in each decoding window of blocks exceeding the capacity. Hence, the last block index in each decoding window is chosen to be  $T_k$, if both $k$ and $T$ are odd numbers or if both $k$ and $T$ are even numbers. Otherwise, it is chosen to be $T_k-1$.  

\item In the basic streaming setup, we note that the error event related to the smallest decoding window dominates the error probability. For alternating message rates, if the smallest decoding window (after the last block index is chosen to be an odd number) consists of an odd number of blocks, the average message rate in that window is away from the capacity. Hence, in that case, we adjust the threshold in the decoding test \eqref{eqn:dec_rule} to ensure that the contribution of the smallest decoding window to the error probability is negligible. 
\end{itemize}
By keeping the above rules in mind, the decoding procedure can be stated as follows:
\subsubsection{$T$ is odd} If $k$ is also an odd number, the decoding procedure is the same as that for the basic streaming setup  in Section \ref{subsec:MD_pf}, except that when $j=k$, we perform the following decoding test instead of \eqref{eqn:dec_rule}:
\begin{align}
i(\mathbf{x}_{[k:T_k]}(\hat{G}_{T_k, [1:k-1]}, g_{[k:T_k]}), \mathbf{y}_{[k:T_k]})&> (T+r-1) \cdot \log M_n. 
\end{align}
If $k$ is an even number, the decoding procedure is the same as that for the basic streaming setup  in Section \ref{subsec:MD_pf}, except that the last block index $T_k$ in each decoding window is replaced by $T_{k}-1$.

\subsubsection{$T$ is even} If $k$ is an odd number,  the decoding procedure is the same as that for the basic streaming setup  in Section \ref{subsec:MD_pf}, except that  the last block index $T_k$ in each decoding window is replaced by $T_{k}-1$ and when $j=k$, we perform the following decoding test: 
\begin{align}
i(\mathbf{x}_{[k:T_k-1]}(\hat{G}_{T_k, [1:k-1]}, g_{[k:T_k-1]}), \mathbf{y}_{[k:T_k-1]})&> (T+r-2) \cdot \log M_n. 
\end{align}
If $k$ is an even number, the decoding procedure is the same as that for the basic streaming setup  in Section \ref{subsec:MD_pf}.

\subsection{Error analysis }
Let us first consider the case where both $T$ and $k$ are odd numbers.  The error event $\{\hat{G}_k\neq G_k\}$ happens only if at least one of the following $2k$ events occurs: 
\begin{align}
\mathcal{E}_{k,j}''&:= \{
i(\mathbf{X}_{[j:T_k]}(G^{T_k}), \mathbf{Y}_{[j:T_k]})\leq (T_k-j+1) \cdot \log M_n
\},~ j\in [1:k-1] \label{eqn:err1_al}\\
\tilde{\mathcal{E}}_{k,j}''&:= \{i(\mathbf{X}_{[j:T_k]}(G^{j-1},g_{[j:T_k]}), \mathbf{Y}_{[j:T_k]})> (T_k-j+1) \cdot \log M_n \cr
&~~~~~~~~~~~~~~~~~~~~~~~~~~~~~~~~~~~~ \mbox{ for some }g_{[j:T_k]} \mbox{ such that }g_{j} \neq G_{j} \},~ j\in [1:k-1] \label{eqn:err2_al}\\
\mathcal{E}_{k,k}''&:= \{
i(\mathbf{X}_{[k:T_k]}(G^{T_k}), \mathbf{Y}_{[k:T_k]})\leq (T+r-1) \cdot \log M_n
\}\\
\tilde{\mathcal{E}}_{k,k}''&:= \{i(\mathbf{X}_{[k:T_k]}(G^{k-1},g_{[k:T_k]}), \mathbf{Y}_{[k:T_k]})> (T+r-1) \cdot \log M_n \cr
&~~~~~~~~~~~~~~~~~~~~~~~~~~~~~~~~~~~~ \mbox{ for some }g_{[k:T_k]} \mbox{ such that }g_{k} \neq G_{k} \}.
\end{align}
Now, we have
\begin{align}
\E_{\mathcal{C}_n}[\Pr(\hat{G}_k\neq G_k|\mathcal{C}_n)]&\leq \sum_{j=1}^{k} \left(\Pr(\mathcal{E}_{k,j}'')+\Pr(\tilde{\mathcal{E}}_{k,j}'') \right).
\end{align}
Fix arbitrary $0<\lambda<1$. By applying similar steps as in the error analysis for the basic streaming setup in Section \ref{subsec:MD_pf}, we can show\footnote{The RHS of \eqref{eqn:MD_sumup_al} can be obtained by replacing $T$ by $T+1$ in the RHS of \eqref{eqn:MD_sumup}.}
\begin{align}
\sum_{j=1}^{k-1} \left(\Pr(\mathcal{E}_{k,j}'')+\Pr(\tilde{\mathcal{E}}_{k,j}'') \right)\
\leq \frac{\exp\left\{-(T+1)n\rho_n^2\lambda^2\left(\frac{1 }{2V}-\lambda \rho_n \tau\right)\right\}}{1-\exp\{-n\rho_n^2\lambda^2\left(\frac{1 }{2V}-\lambda \rho_n \tau\right)\}}+\frac{\exp\left\{-(T+1)n(1-\lambda)\rho_n\right\}}{1-\exp\left\{-n(1-\lambda)\rho_n\right\}} \label{eqn:MD_sumup_al}
\end{align}
for sufficiently large $n$,  where $\tau$ is some non-negative constant dependent only on the input distribution $P_X(x)$ and channel statistics $W(y|x)$. 

Now, $\Pr(\mathcal{E}_{k,k}'')+\Pr(\tilde{\mathcal{E}}_{k,k}'')$ is bounded as follows:
\begin{align}
&\Pr(\mathcal{E}_{k,k}'')+\Pr(\tilde{\mathcal{E}}_{k,k}'')\cr
&\leq \Pr\left(\sum_{l=1}^{nT} i(X_l;Y_l)\leq (T+r-1) \cdot \log M_n\right)\cr
&\qquad\qquad\qquad\qquad+M_n^{T+r-1} \Pr\left(\sum_{l=1}^{nT} i(X_l;\bar{Y}_l)> (T+r-1) \cdot \log M_n\right)\\
&\overset{(a)}{=}\E\left[\exp\left\{-\left[\sum_{l=1}^{nT} i(X_l;Y_l)- (T+r-1)\cdot \log M_n \right]^+\right\}\right]\\
&\overset{(b)}{\leq} \Pr\left(\sum_{l=1}^{nT} i(X_l;Y_l)\leq (T+r-1)n(C-\lambda\rho_n)\right)+\exp\left\{-(T+r-1)n(1-\lambda)\rho_n\right\}\\
&\leq \Pr\left(\sum_{l=1}^{nT} i(X_l;Y_l)\leq (T+r-1)nC\right)+\exp\left\{-(T+r-1)n(1-\lambda)\rho_n\right\}\\
&\overset{(c)}{\leq}\exp\left\{-Tn\tau'\right\}+\exp\left\{-(T+r-1)n(1-\lambda)\rho_n\right\}
\end{align}
for some constant $\tau'>0$ that depends only on the input distribution $P_X(x)$ and channel statistics $W(y|x)$, where $(X_l,Y_l, \bar{Y}_l)$'s are i.i.d. random variables each generated according to $P_X(x_l)W(y_l|x_l)P_XW(\bar{y}_l)$, $(a)$ is from the identity \cite[Eq.~(69)]{PolyanskiyPoorVerdu:10}, $(b)$ follows from similar steps as  \eqref{eqn:fixj}-\eqref{eqn:spr2}, and $(c)$ is from Lemma \ref{lemma:central} that is relegated to the end of this appendix. 

Thus, we obtain 
\begin{align}
\E_{\mathcal{C}_n}[\Pr(\hat{G}_k\neq G_k|\mathcal{C}_n)]&\leq\frac{\exp\left\{-(T+1)n\rho_n^2\lambda^2\left(\frac{1 }{2V}-\lambda \rho_n \tau\right)\right\}}{1-\exp\{-n\rho_n^2\lambda^2\left(\frac{1 }{2V}-\lambda \rho_n \tau\right)\}}+\frac{\exp\left\{-(T+1)n(1-\lambda)\rho_n\right\}}{1-\exp\left\{-n(1-\lambda)\rho_n\right\}}\cr
&\qquad\qquad+\exp\left\{-Tn\tau'\right\}+\exp\left\{-(T+r-1)n(1-\lambda)\rho_n\right\}
\end{align}
when both $T$ and $k$ are odd numbers. 

Now, let us consider the case where $T$ is an odd number and $k$ is an even number. We remind that the decoding procedure is the same as that for the basic streaming setup  in Section \ref{subsec:MD_pf}, except that $T_k$ is replaced by $T_{k}-1$. By applying similar steps as in the error analysis for the basic streaming setup in Section \ref{subsec:MD_pf}, we can show\footnote{The RHS of \eqref{eqn:MD_sumup_al23} can be obtained by replacing $T$ by $T-1$ in the RHS of \eqref{eqn:MD_sumup}.}
\begin{align}
\E_{\mathcal{C}_n}[\Pr(\hat{G}_k\neq G_k|\mathcal{C}_n)]&\leq \frac{\exp\left\{-(T-1)n\rho_n^2\lambda^2\left(\frac{1 }{2V}-\lambda \rho_n \tau\right)\right\}}{1-\exp\{-n\rho_n^2\lambda^2\left(\frac{1 }{2V}-\lambda \rho_n \tau\right)\}}+\frac{\exp\left\{-(T-1)n(1-\lambda)\rho_n\right\}}{1-\exp\left\{-n(1-\lambda)\rho_n\right\}} \label{eqn:MD_sumup_al23}
\end{align}
for sufficiently large $n$.

By following similar statements using Markov's inequality in the proof of Theorem \ref{thm:erasure}, we can show that there must exist a sequence of codes $\mathcal{C}_n$ that satisfies 
\begin{align}
\limsup_{n\rightarrow \infty} \frac{1}{n\rho_n^2 }\log \left(\limsup_{N\rightarrow \infty}\sum_{j=1}^N\frac{\Pr(\hat{G}_{2j-1}\neq G_{2j-1}|\mathcal{C}_n)}{N}\right) \leq\frac{(T+1)\lambda^2}{2V}
\end{align}
and 
\begin{align}
\limsup_{n\rightarrow \infty} \frac{1}{n\rho_n^2 }\log \left(\limsup_{N\rightarrow \infty}\sum_{j=1}^N\frac{\Pr(\hat{G}_{2j}\neq G_{2j}|\mathcal{C}_n)}{N}\right) \leq\frac{(T-1)\lambda^2}{2V},
\end{align}
which completes the proof for the case where $T$ is an odd number, by taking $\lambda\rightarrow 1$. The proof for the case where $T$ is an even number can be done in a similar manner, and hence it is omitted. \endproof

The following lemma used in the proof of Theorem \ref{thm:alternate} corresponds to a non-asymptotic upper bound of the Cram\'{e}r's theorem.
\begin{lemma} \label{lemma:central}
Let $\{Z_l\}_{l\geq 1}$ be a sequence of i.i.d.  random variables with zero mean such that its cumulant generating function $h(s) := \log\E[\exp\{sZ_1\}]$ for $s\geq 0$ is continuously differentiable. For $\eps>0$, the following bound holds: 
\begin{align}
\Pr\left(\frac{1}{n}\sum_{l=1}^n Z_l \geq \eps\right) &\leq \exp\{-nI(\eps)\},
\end{align}
where $I(\eps)$ is the {\em rate function} defined as follows:
\begin{equation}
I(\eps) :=\sup_{s\ge 0}   \{ s\eps -h(s) \}>0.
\end{equation}
\end{lemma}
\begin{proof}
For any $s\ge 0$, we obtain
\begin{align}
\Pr\bigg(\frac{1}{n}\sum_{l=1}^n Z_l \ge\eps \bigg)  & \leq  \exp\big\{  - n ( s\eps - h(s)  ) \big\}
\end{align}
by applying the same steps used to obtain \eqref{eqn:chernoff} in the proof of Lemma  \ref{lemma:nonasymptotic_MD}.
Since $s\ge 0$ is arbitrary, we obtain the following bound
\begin{equation}
\Pr\bigg(\frac{1}{n}\sum_{l=1}^n Z_l \geq \eps \bigg)  \le\exp\{-n I(\eps) \}.
\end{equation}
Furthermore, because  $s\eps-h(s)|_{s=0}=0$ and $\frac{d(s\eps-h(s))}{ds}|_{s=0}=\eps>0$, we conclude that $I(\eps)>0$. 
\end{proof}

\end{document}